\documentclass[aap]{imsart}

\RequirePackage{amsthm,amsmath,amsfonts,amssymb}
\RequirePackage[numbers,sort&compress]{natbib}
\RequirePackage[colorlinks,citecolor=blue,urlcolor=blue]{hyperref}
\RequirePackage{graphicx}

\usepackage{latexsym,upgreek}
\usepackage{amsfonts}
\usepackage{mathrsfs}
\usepackage{amsmath}
\usepackage{amsthm}
\usepackage{graphicx} 
\usepackage{amssymb}
\usepackage{stmaryrd}
\usepackage{natbib}
\usepackage{bbm}
\usepackage{commath}
\usepackage{lipsum}
\usepackage{float} 
\usepackage{enumitem}
\usepackage{comment} 
\usepackage{cancel}

\startlocaldefs

\newcommand{\rw}{\rightarrow}    
\newcommand{\Real}{\mathbb{R}}

\newcommand{\mB}{{\mathcal B}}
\newcommand{\mP}{{\mathcal P}}

\newcommand{\mS}{{\mathcal S}}
\newcommand{\mI}{{\mathcal I}}

\newcommand{\mG}{{\mathcal G}}
\newcommand{\mF}{{\mathcal F}}

\newcommand{\mX}{{\mathcal X}}
\newcommand{\mY}{{\mathcal Y}}
\newcommand{\mN}{{\mathcal N}}
\newcommand{\mM}{{\mathcal M}}

\newcommand{\mO}{{\mathcal O}}

\newcommand{\mZ}{{\mathcal Z}}

\newcommand{\mbN}{\mathbb{N}}

\newcommand{\mbP}{\mathbb{P}}
\newcommand{\mbR}{\mathbb{R}}
 
\newcommand{\mbE}{\mathbb{E}} 

\newcommand{\sff}{{\sf f}}

\newcommand{\sd}{{\sf d}}

\newcommand{\sfd}{{\sf d}}

\newcommand{\mix}{\mathfrak{m}}

\newcommand{\beq}{\begin{equation}}
\newcommand{\eeq}{\end{equation}}
\newcommand{\beqa}{\begin{eqnarray}}
\newcommand{\eeqa}{\end{eqnarray}}
\newcommand{\nn}{\nonumber}

\theoremstyle{plain}

\newtheorem{theorem}{Theorem}[section]
\newtheorem{lemma}[theorem]{Lemma}
\theoremstyle{definition}

\newtheorem{proposition}[theorem]{Proposition}
\newtheorem{corollary}[theorem]{Corollary}

\newtheorem{remark}{Remark}
\newtheorem{algorithm}{Algorithm}

\endlocaldefs

\begin{document}

\begin{frontmatter}
\title{Nested importance sampling for Bayesian inference: error bounds and the role of dimension}
\runtitle{Nested importance sampling for Bayesian inference}

\begin{aug}
\author[A]{\fnms{Fabián}~\snm{González} \thanks{[\textbf{Corresponding author.}]}\ead[label=e1]{omgonzal@math.uc3m.es}\orcid{0000-0003-4353-5737}},
\author[B]{\fnms{Víctor}~\snm{Elvira}\ead[label=e2]{victor.elvira@ed.ac.uk}\orcid{0000-0002-8967-4866}}
\and
\author[A]{\fnms{Joaquín}~\snm{Míguez}\ead[label=e3]{joaquin.miguez@uc3m.es}\orcid{0000-0001-5227-7253}}
\address[A]{Department of Signal Theory and Communications, Universidad Carlos III de Madrid, Spain \printead[presep={ ,\ }]{e1,e3}}

\address[B]{School of Mathematics, University of Edinburgh, United Kingdom \printead[presep={,\ }]{e2}}
\end{aug}

\begin{abstract}
Many Bayesian inference problems involve high dimensional models for which only a subset of the model variables are of actual interest. All other variables are just nuisance parameters that one would ideally like to integrate out analytically. Unfortunately, such integration is often impossible. There are several computational methods that have been proposed over the past 15 years that replace intractable analytical marginalization by numerical integration, typically using different flavours of importance sampling (IS). Such methods include particle Markov chain Monte Carlo, sequential Monte Carlo squared (SMC$^2$), IS$^2$, nested particle filters, and others. In this paper, we investigate the role of the dimension of the nuisance variables in the error bounds achieved by nested IS methods in Bayesian inference. We prove that, under suitable regularity assumptions on the model, the approximation errors increase at a polynomial (rather than exponential) rate w.r.t. the dimension of the nuisance variables. Our analysis relies on tools from functional analysis and measure theory, and it includes the case of polynomials of degree zero, where the approximation error remains uniformly bounded, even as the dimension of the nuisance variables grows without bound. We also show how the general analysis can be applied to specific classes of models, including linear and Gaussian settings, models with bounded observation functions, and others. These findings improve the current understanding of how the performance of IS techniques scales with dimension in Bayesian inference problems.
\end{abstract}



\end{frontmatter}


\section{Introduction}\label{sIntroduction}

Bayesian inference is a cornerstone of modern statistical methodology, widely used in fields such as signal, image processing, and machine learning. Many problems in these domains are naturally framed in terms of high-dimensional probabilistic models, where solutions are given as integrals. Unfortunately, this integration can only be performed analytically for a limited set of models for which the prior distribution and the likelihood function are specifically selected to yield a closed-form posterior distribution, as it is the case for a handful of conjugate families \cite{Bernardo94}. The obvious alternative is numerical integration \cite{Robert04}. One widely used approach for handling intractable integrals in Bayesian inference is importance sampling (IS) \cite{Robert04,tokdar2010importance,Bugallo17}.


\subsection{Motivation and background}
 
 IS is a Monte Carlo integration technique that employs a collection of random samples $X=x^i$, $i=1, \ldots, N$, drawn from a proposal probability measure $\tau$ to approximate integrals with respect to (w.r.t.) a target probability measure $\pi$. Each sample $x^i$ is assigned a weight $w^i$, which is 
\begin{itemize}
\item proportional to the relative density $\frac{\sfd \pi}{\sfd \tau}(x^i)$,
\item and normalised, $\sum_{i=1}^N w^i = 1$.
\end{itemize}
Under mild assumptions weighted averages can be proved to converge to the corresponding integrals \cite{Robert04,Chopin20}, i.e., 
$$
\sum_{i=1}^N w^i \varphi(x^i) \stackrel{N\rw\infty}{\longrightarrow} \int \varphi(x)\pi(\sfd x),
$$ 
where $\varphi(\cdot)$ is a test function. 

IS is a classic methodology with a plethora of applications \cite{Doucet01b,Robert04,Bucklew04,DelMoral06,Bugallo17}. It provides a flexible and easy-to-use method for approximating expectations w.r.t. complex probability distributions.
Over the past decade, several advanced IS-based methods, \cite{Andrieu10,Chopin12,Crisan18bernoulli,Tran21,doucet2022score,doucet2022annealed,doucet2018sequential,akyildiz2021convergence} have been developed to extend importance sampling and handle inference in increasingly sophisticated models, either static or dynamic. See also \cite{agapiou2017importance,chatterjee2018sample} for a recent account of the fundamentals of IS. 
 
IS-based methods are often applied in Bayesian inference settings, where the target probability measure $\pi$ is the posterior distribution of a random variable $X$ (sometimes termed the state variable, or simply the state) conditional on some random variable (r.v.) $Y$ that represents the observations available in the system of interest. To be specific, given a realization $Y=y$, the objective is to determine the conditional probability law of $X$ given $y$. This is formulated via Bayes' theorem
$$
\pi(\sd x) = \frac{l_y(x)\pi_0(\sd x)}{\int l_y(x')\pi_0(\sd x')} = \frac{l_y(x)\pi_0(\sd x)}{\pi_0(l_y)}.
$$
where, $l_y(x)$ denotes the likelihood of the observation $Y=y$ given $X=x$, $\pi_0(\sd x)$ is the prior probability measure of the state $X$, and $\pi_0(l_y)=\int l_y(x')\pi_0(\sd x')$ is the marginal likelihood, or model evidence, which acts as the normalization constant for the posterior distribution $\pi(\sd x)$.

\subsection{Impact of dimension}

Despite its versatility, IS is known to suffer from limitations, most notably the so-called \textit{curse of dimensionality} \cite{agapiou2017importance,chatterjee2018sample,Rebeschini15,Snyder08}. Theoretical and empirical studies suggest that performance can deteriorate significantly when inference involves high-dimensional state spaces  (where the number of scalar r.v.'s or parameters to be inferred is referred to as the dimension of the state space).  
For example, in~\cite{chatterjee2018sample} the authors show that, in general, the number of samples required to approximate the posterior probability law $\pi(\sd x)$ (within a specified accuracy) can scale exponentially with the Kullback--Leibler (KL) divergence between the posterior and the proposal distributions. Also, importance samplers can suffer from weight degeneracy when the dimension of the state space is large \cite{Bengtsson08,Snyder08,Snyder15}. While sophisticated methods have been proposed to overcome this difficulty in various setups \cite{VanLeeuwen13,Beskos14,Koblents15,Rebeschini15,Beskos17,Ruzayqat22}, they are often computationally costly and rely on assumptions which are hard to assess.

Recent works have begun to characterise the multiple roles of dimension in the analysis of IS algorithms. In particular, the authors of ~\cite{agapiou2017importance} have identified a key quantity governing the efficiency of IS estimators: the second moment of the Radon--Nikodym derivative of the posterior w.r.t. the prior measure. This quantity is further related to important metrics such as the effective sample size (ESS) \cite{Martino17}, and it also appears in bounds on the mean squared error of IS approximations \cite{akyildiz2021convergence}. It plays a central role in our analysis throughout this paper as well.

Other works have demonstrated that the degradation of performance in high-dimensional settings can be alleviated under certain conditions. Adaptive and structure-aware variants of IS (such as those found in the adaptive particle filtering literature~\cite{VanLeeuwen13,Rebeschini15,Beskos17,Ruzayqat22,kuntz2024divide} and in sequential Monte Carlo samplers~\cite{Beskos14}) exploit problem-specific geometric features or sequential data structure to maintain accuracy in moderately high dimensional problems.

\subsection{Nuisance variables}
 There are many problems that require complex, high dimensional models for their proper description, yet the state variables of actual interest for inference form just a (possibly small) subset of the total (possibly large) set of variables in the model. 
  Whenever possible, it is desirable to identify the subset of variables of interest, while integrating out any other {\em nuisance} random quantities. When inference is restricted to a low-dimensional projection of the latent
space, numerical marginalization of nuisance variables can, in principle, reduce the effective complexity of the problem and improve accuracy of estimators.

 Several algorithms have been proposed over the past 15 years for this purpose, including numerical Rao-Blackwellisation for particle filters \cite{Hu08,Johansen12,Naesseth19,Rebeschini15}, particle Markov chain Monte Carlo methods \cite{Andrieu10}, nested IS schemes \cite{Tran21,ashton2022nested,naesseth2015nested}, and sequential importance samplers for state-space dynamical systems, including sequential Monte Carlo square (SMC$^2$) \cite{golightly2018efficient,Chopin12} and nested particle filters \cite{Crisan18bernoulli,Crisan17}.

Nested IS methods have been popular for Bayesian inference in a variety of 
applications \cite{ashton2022nested,lange2023nautilus,pullen2014bayesian}. The basic scheme relies on computing unbiased estimates of the marginal likelihood of the parameters of interest to approximate their posterior distribution given the observed data, while numerically integrating out all other variables. In \cite{Tran21}, asymptotic convergence results are given for nested IS estimators, including a central limit theorem.  
The key to the analysis is the realisation that unbiased estimators of the marginal likelihood of the parameters of interest can be computed with a fixed and finite number of Monte Carlo samples \cite{Chopin12,Tran21} --the same idea underlying, e.g., particle Markov chain Monte Carlo methods \cite{Andrieu09,Andrieu10}. In general, the variance of these estimators can be expected to increase with the dimension of the nuisance variables. 


\subsection{Contributions}

In this paper, we investigate the role of the dimension of the nuisance variables (denoted $d_z$) in the error bounds achieved by nested IS methods in Bayesian inference problems\footnote{We focus on Bayesian inference for static models, rather than Bayesian filtering in dynamical systems, but we expect that many results can be extended to the latter setting and this will be the focus of future work.}. We begin with a (rather classical) derivation of the error bounds, from where we explicitly identify their dependence on $d_z$. The key factors to analyze are the marginal likelihood of the model and, more specifically, the (generally intractable) likelihood of the target variables that results from integrating out the nuisance variables. In order to gain some insight, we first look into the class of linear and Gaussian models. For this case, we find sufficient conditions (expressed in terms of the spectral decomposition of the model matrices) that guarantee that the approximation error bounds of the nested importance sampler increase at a polynomial (rather than exponential) rate w.r.t. the dimension of the nuisance variables, $d_z$. We show how this result can be attained either assuming deterministic (arbitrary) observations {\em or} assuming that the observations are random and generated by the model. Remarkably, in the latter case the analysis can be simplified.

Following the insight gained from the linear and Gaussian setting, we proceed to analyze the $L^1$ and $L^2$ error bounds for general models assuming random observations, which yield a random posterior measure and random error bounds as well. Within this framework, we first obtain sufficient and necessary conditions for the approximation errors of the nested importance sampler to be finite and converge with $\mO\left(N^{-\frac{1}{2}}\right)$ for any given (arbitrary large) dimension $d_z$, where $N$ is the number of Monte Carlo samples. Our analysis relies on tools from measure theory and functional analysis, namely the notion of Bochner integrability. Then, we proceed to obtain sufficient conditions for the approximation errors to increase polynomially (rather than exponentially) with the dimension of the nuisance variables $d_z$, similar to the linear and Gaussian case. We explicitly highlight the case where the polynomial degree is 0, which corresponds to error bounds that hold uniformly over $d_z$. We apply these general results to several examples, including linear systems, models with bounded observation functions, and models with pointwise-bounded likelihood. In all cases, we provide analytical proof of the validity of our bounds. 
Finally we show that the proposed methodology can also be used to quantify the role of the dimension of the signal of interest in standard IS schemes (with no nuisance variables). In particular we identify sufficient conditions for the error bounds of a conventional IS algorithm to increase polynomially, or even remain constant, w.r.t. the dimension of the variable of interest. 

This analysis improves our current understanding of how the performance of IS algorithms scales with the dimension of probabilistic models in Bayesian inference problems.

\subsection{Outline of the Paper}

The paper is organized as follows. Section~\ref{sec:sIS2} formally describes the probabilistic model and the nested importance sampling algorithm.  In Section~\ref{sec:Linear_Gaussian}, we analyze the error bounds attained for the linear and Gaussian model and then investigate a broad class of possibly nonlinear and non-Gaussian models in Section~\ref{sec:General_Random}. The general results obtained in that framework are then particularised to some examples in Section~\ref{sec:Special_cases}. Section~\ref{sec:standard_random} demonstrates how we can derive analogous results in the standard IS framework (without nuisance variables). Finally, in Section~\ref{sec:Discussion}, we summarise our main findings and discuss some connections with existing results, as well as possible venues for future research. Technical proofs and some additional results are provided in the appendices. We conclude this introduction with a brief summary of the notation used throughout the paper. 

\subsection{Summary of notation} \label{ssNotation}
\begin{itemize}
    \item Sets, measures, and integrals:
    \begin{itemize}
        \item[-] $\mB (S)$ is the $\sigma$-algebra of Borel subsets of $S \subseteq  \mbR^d$.
\item[-] $\mP(S) := \{ \nu  : \mB (S) \mapsto  [0, 1]$ and $\nu (S) = 1\}$  is the set of probability measures over
$\mB (S)$.
\item[-] $ \nu (f) := 
\int f d\nu$  is the integral of a Borel measurable function $f : S \mapsto  \mbR$ w.r.t. the
measure $\nu  \in  \mP (S)$.

\end{itemize}
\item Functions:
    \begin{itemize}
    
    \item [-] Consider the measurable spaces \((S, \mathcal{B}(S))\) and \((\mathbb{R}, \mathcal{B}(\mathbb{R}))\).  
We denote by \(B(S)\) the space of bounded, real-valued, measurable functions \(f : S \mapsto \mathbb{R}\).  
For any \(f \in B(S)\), the uniform norm is defined as
\[
\|f\|_{\infty} := \sup_{s \in S} |f(s)| < \infty.
\]
    \end{itemize}
    
\item Real r.v.'s on a probability space $(\Omega,\mF,\mbP)$ are denoted by capital letters (e.g., $Z:\Omega \mapsto \mbR^d)$, while their realisations are written as lowercase letters (e.g., $Z(\omega)=z$, or simply, $Z=z$). If $X$ is a $d_x$ multivariate Gaussian r.v., then its probability law is denoted $\mN(\sd x;m, C)$, where $m$ is the mean and $C$ is the covariance matrix. If the r.v. $X$ has probability law $\pi$ and $f$ is a measurable function, then we denote $$\pi(f) = \int f \sd\pi = \mbE[f(X)],$$
where $\mbE[\cdot]$ is the expectation operator.

\item Linear algebra:
    \begin{itemize}
       \item Let ${A}\in \mathbb{R}^{m\times p}$  be a given matrix, we  denote the ordered singular values of ${A}$ as  $$\sigma_1({A})\geq\dots \geq \sigma_{\min\{m,p\}}({A}).$$
        \item Let $A \in \mathbb{R}^{n \times n}$ be a real symmetric matrix. The \emph{spectrum} of $A$ is denoted as 
\[
\mathrm{spec}(A) := \{ \lambda \in \mathbb{R} \, : \, \exists\, v \in \mathbb{R}^n \setminus \{0\} \text{ such that } Av = \lambda v \}.
\]
\item We denote the real eigenvalues of $A$ as
$\lambda_1(A) \geq \lambda_2(A) \geq \cdots \geq \lambda_n(A),$
where $\lambda_1(A)$ and $\lambda_n(A)$ denote the maximum and minimum eigenvalues of $A$, respectively.
   \end{itemize}
   
\item Functional analysis:
    \begin{itemize}
        \item For $1\leq p < \infty$, and $\nu \in \mP(S),$ denote by $
L^p(\nu)$ the Banach space defined by  $$
L^p(\nu)= \left\{\,f\colon S \mapsto \mathbb{R} : \|f\|_p := \left(\int_S |f(s)|^p\,\sd \nu(s)\right)^{1/p} < \infty \right\}.$$
        \item In particular, for $p=2$, the inner product in the Hilbert space $L^2(\nu)$ is denoted by $$\langle \varphi, \psi \rangle_{L^2(\nu)}:=\int_S\varphi(s) \psi(s) \sd \nu(s).$$
    \end{itemize}

\end{itemize}


\section{Nested importance samplers} \label{sec:sIS2}


\subsection{Model} \label{ssec:sIS2}
Let $(\Omega,\Sigma,\mbP)$ be a probability space and let $X:\Omega\mapsto\mX \subseteq \Real^{d_x}$, $Z:\Omega\mapsto\mZ \subseteq \Real^{d_z}$, and $Y:\Omega\mapsto\mY \subseteq \Real^{d_y}$ denote three multidimensional r.v.'s. To be specific, the r.v. $X$ represents the $d_x$-dimensional signal of interest (or state), which we aim to estimate, $Y$ represents some $d_y$-dimensional observed data, and $Z$ is a $d_z$-dimensional nuisance variable needed to model the relationship between $X$ and $Y$.

The complete model can be specified as follows:
\begin{itemize}
\item The prior probability law of the state of interest $X$ is denoted by $\pi_0(\mathrm{d}x)$, and we assume the ability to generate random samples from this distribution.

\item The nuisance variable $Z$ and the state $X$ are related through a Markov kernel $\kappa : \mathcal{X} \times \mathcal{B}(\mathcal{Z}) \to [0, 1]$, where $\mathcal{B}(\mathcal{Z})$ denotes the Borel $\sigma$-algebra of subsets of $\mathcal{Z}$.

\item Given $X = x$ and $Z = z$, the observation $Y$ admits a conditional probability density function (pdf) w.r.t. the Lebesgue measure, denoted by $g(y \mid x, z)$.

\end{itemize}

We refer to a triple $\mathcal{M} = (\pi_0, \kappa, g)$ as a model for Bayesian inference. Throughout this work, we consider a family of models $\{\mathcal{M}^{d_z}\}_{d_z \in \mathbb{N}} = \{\pi_0, \kappa^{d_z}, g^{d_z}\}_{d_z \in \mbN}$, indexed by the dimension $d_z$ of the nuisance variable $Z$. The superscript $d_z$ explicitly denotes the potential dependence of the kernel $\kappa^{d_z}$ and the likelihood $g^{d_z}$ on the dimension of $Z$. Importantly, these models are of the same fundamental type; their differences solely arise from varying the dimensionality of the nuisance variable $Z$. This implies that while $d_z$ can increase, the inherent structure of the transition kernel $\kappa^{d_z}$ and the likelihood $g^{d_z}$ remains consistent. For an explicit example, see Section \ref{sec:Linear_Gaussian}. Similar superscript notation is  used for other model-dependent objects throughout this paper.

We introduce the notation $g_y^{d_z}(x, z)$ for an associated likelihood function, defined as
\[
g_y^{d_z}(x, z) := c \, g^{d_z}(y \mid x, z) > 0,
\]
where $c \in \mathbb{R}^+$ is a positive (possibly unknown) constant, independent of $x$ and $z$. The likelihood of a state value $X = x$ given an observation $Y = y$ for this model is
\begin{equation}
l_y^{d_z}(x) := \int g_y^{d_z}(x, z) \, \kappa^{d_z}(x, \sd z).
\label{eq_ly}
\end{equation}
Note, however, that the integral above cannot be computed exactly for many practical models of interest.

Our aim is to approximate the posterior law of the state $X$ conditional on the observation $Y=y$, denoted $\pi(\sd x)$, using IS. This posterior probability measure can be written as
\beq
\pi(\sd x) = \frac{
	l_y^{d_z}(x) \pi_0(\sd x)
}{
	\pi_0(l_y^{d_z})
},
\nn
\eeq
where the likelihood $l_y^{d_z}(x)$ is given by Eq. \eqref{eq_ly} and 
\beq 
\pi_0(l_y^{d_z}):=\int l_y^{d_z}(x) \pi_0(\sd x)
\nn
\eeq 
is the normalization constant of the distribution, for $y\in \mY$ fixed, and also corresponds to the marginal density of the r.v. \( Y \). We remark that both $l_y^{d_z}(x)$ and $\pi_0(l_y^{d_z})$ are analytically intractable in general. One remarkable  exception is the case of of linear Gaussian models, discussed in Section \ref{sec:Linear_Gaussian}. 

Let us define the joint measure \(\mix^{d_z}\) on \( \mB(\mX \times \mZ)\)  constructed as

\begin{equation}\label{eq:Mix_meas}
\mix^{d_z}( \sd (x, z)) :=  \kappa^{d_z}(x, \sd z) \, \pi_0(\sd x).
\end{equation}
This probability measure represents the joint distribution of the pair \((X, Z)\), where \(X \sim \pi_0\) and \(Z \sim \kappa^{d_z}(x, \cdot)\). We assume that $g_y^{d_z}(x,z) \in L^2(\mix^{d_z}),$ $\forall \; y \in \mY$, i.e.,
\begin{equation}\label{eq:normL2}
\| g_y \|^2_{L^2(\mix^{d_z})}:=\int_{\mathcal{X}} \int_{\mathcal{Z}} g_y^2(x,z) \, \kappa^{d_z}(x, \sd z) \, \pi_0(\sd x) 
 < \infty,
\end{equation} something that holds trivially when $g_y^{d_z}(x,z)$ is uniformly bounded for all $ (x,z)\in \mX \times \mZ$.

\begin{remark}
Throughout this section and in Section~\ref{sec:Linear_Gaussian}, we assume that the observation $y \in \mathcal{Y}$ is arbitrary but fixed. Moreover, we assume that the likelihood $g_y^{d_z}(x,z)$ is positive and bounded for all $y \in \mathcal{Y}$, i.e.,
$$\|g_y^{d_z}\|_{\infty}=\sup_{(x,z)\in \mX \times \mZ} \lvert g_{y}^{d_z}(x,z) \rvert< \infty.$$ 
\end{remark}

The approximation of the posterior law $\pi(\sd x)$ for a given (fixed) observation $Y=y$ can be carried out using a simple nested IS scheme as shown in Section \ref{Algorithm}.



\subsection{Nested importance sampler}\label{Algorithm}

Assume a model $\mM^{d_z}$ as described in Section \ref{ssec:sIS2} above. In order to describe a general nested importance sampler for the approximation of $\pi(\sd x)$, let us introduce an importance probability measure $\nu$ such that $\pi_0$ is absolutely continuous w.r.t. $\nu$ (denoted $\pi_0 \ll \nu$) and a Markov kernel $\tau(x,\sd z)$ such that $ \kappa^{d_z}(x,\cdot) \ll \tau(x,\cdot)$ almost everywhere on $\mX$. Let $\frac{\sd \pi_0}{\sd \nu}(x)$ denote the relative density of $\pi_0$ w.r.t. $\nu$ evaluated at $x\in \mX$ and let $\frac{\sd \kappa^{d_z}}{\sd \tau}(x,z)$ denote the relative density of $\kappa^{d_z}(x,\sd z)$ w.r.t. $\tau(x,\sd z)$ evaluated at $z \in \mZ$, for given $x$.

\begin{algorithm}\label{A1}
A general nested importance sampler

\begin{enumerate}
\item Draw $N$ i.i.d. samples $x^1, \ldots, x^N$ with common probability law $\nu$.
\item For each $i=1, \ldots, N$, 
	\begin{itemize}
	\item draw $M$ i.i.d. samples $z^{i,1},\ldots, z^{i,M}$ with common probability law $\tau(x^i,\sd z)$, and
	\item compute the approximate likelihood 
	\beq
	l_{y,\tau}^{M,d_z}(x^i) = \frac{1}{M}\sum_{j=1}^M \frac{\sd \kappa^{d_z}}{\sd \tau}(x^i,z^{ij}) g_y^{d_z}(x^i,z^{ij}).
	\nn
	\eeq
	\end{itemize}
\item Compute normalised importance weights $w_\nu^{i,M} \propto l_{y,\tau}^{M,d_z}(x^i) \frac{\sd \pi_0}{\sd \nu}(x^i)$.
\end{enumerate} 
\end{algorithm}

The samples $x^1, \ldots, x^N$ and their importance weights yield a Monte Carlo estimator of the posterior measure $\pi(\sd x)$, namely
\beq
\pi_{\nu,\tau}^{N,M}(\sd x) := \sum_{i=1}^N w_\nu^{i,M} \delta_{x^i}(\sd x),
\nn
\eeq
where $\delta_{x^i}(\sd x)$ is the Dirac delta measure located at $x^i$.

The scheme above requires the ability to sample and evaluate, point-wise, the relative densities $\frac{\sd \pi_0}{\sd \nu}(x)$ and $\frac{\sd \kappa^{d_z}}{\sd \tau}(x,z)$ \cite{agapiou2017importance}. For clarity, we carry out our analysis for a simpler version of this general algorithm where $\nu=\pi_0$, $\tau=\kappa$, and we only need to assume the ability to sample (not evaluate) the prior $\pi_0$ and the kernel $\kappa^{d_z}(x,\cdot)$. The resulting scheme is displayed as Algorithm \ref{A2}. Hereafter we refer to this specific procedure as nested IS, although other implementations are obviously possible with different choices of $\nu$ and $\tau$.

\begin{algorithm} \label{A2}
Standard nested importance sampler 
\begin{enumerate}
\item Draw $N$ i.i.d. samples $x^1, \ldots, x^N$ from the prior law $\pi_0$.
\item For each $i=1, \ldots, N$, 
	\begin{itemize}
	\item draw $M$ i.i.d. samples $z^{i,1},\ldots, z^{i,M}$ using the kernel $\kappa^{d_z}(x^i,\sd z)$, and
	\item compute the approximate likelihood 
	$
	l_y^{M,d_z}(x^i) = \frac{1}{M}\sum_{j=1}^M g_y^{d_z}(x^i,z^{ij}).
	$
	\end{itemize}
\item Compute normalised importance weights $w^{i,M} = \frac{l_y^{M,d_z}(x^i)}{\sum_{j=1}^N l_y^{M,d_z}(x^j)}$.
\item Output the random probability measure $\pi^{N,M}(\sd x) = \sum_{i=1}^N w^{i,M} \delta_{x^i}(\sd x)$
\end{enumerate}
\end{algorithm}

\begin{remark}
    Our analysis of Algorithm \ref{A2},  including Theorem \ref{thIS2} and the results in Sections \ref{sec:Linear_Gaussian}-\ref{sec:Special_cases}, can be extended in a straightforward manner to the more general Algorithm \ref{A1}. This is explicitly discussed in Appendix \ref{sec:importance_sampling}.
\end{remark}

The measure $\pi^{N,M}$ can be easily used to approximate posterior expectations of $X$, i.e., integrals w.r.t. $\pi(\sd x)$. For a given real test function $f:\mX \mapsto \Real$, let us denote
\beq
\pi(f) := \int_\mX f(x)\pi(\sd x)=\mbE[f(X) \mid Y=y],
\nn
\eeq
hence $\pi(f)$ is the expected value of the r.v. $f(X)$ conditional on $Y=y$. We can naturally approximate
\beq
\pi(f) \approx \pi^{N,M}(f) = \int_\mX f(x)\pi^{N,M}(\sd x) = \sum_{i=1}^N f(x^i) w^{i,M},
\nn
\eeq
and then analyze the random errors $\pi(f)-\pi^{N,M}(f)$. In particular, the next section is devoted to the calculation of upper bounds for the $L_p$ norms of the error $\pi(f)-\pi^{N,M}(f)$, namely
\beq
\left\|
	\pi(f)-\pi^{N,M}(f)
\right\|_p = \mbE\left[
	\left|
		\pi(f)-\pi^{N,M}(f)
	\right|^p
\right]^{\frac{1}{p}}, \quad p \ge 1,
\nn
\eeq
where the real test function $f$ is assumed to belong to $B(\mX)$, the space of real‑valued, measurable, bounded functions on $(\mX,\mB(\mX))$. Note that the expectation $\mbE[\cdot]$ in the expression above is computed w.r.t. the distribution of the measure-valued r.v. $\pi^{N,M}$ generated by the nested IS Algorithm \ref{A2}.


\subsection{Error bounds}\label{sec:EB}
 Theorem \ref{thIS2} below provides an explicit error bound for the estimator $\pi^{N,M}(f)$ for bounded test functions, $f \in B(\mX)$, under mild assumptions on the model of Section \ref{ssec:sIS2}.

\begin{theorem} \label{thIS2}
For given $Y=y$, assume that the likelihood function \( g_y^{d_z}(x,z) \) is strictly positive and bounded, i.e.,
\[
g_y^{d_z}(x,z) > 0 \quad \text{for all } (x,z) \in \mathcal{X} \times \mathcal{Z}, \quad \text{and} \quad \|g_y^{d_z}\|_{\infty} < \infty.
\]
Then, there exists a constant \( C_y^{d_z} < \infty \), independent of \( N \), \( M \), and \( f \), such that for any \( f \in B(\mathcal{X}) \) and any \( p \geq 1 \), the following inequality holds
\begin{equation}\label{eq:ThIS2}
\left\| \pi(f) - \pi^{N,M}(f) \right\|_p \leq \frac{C_y^{d_z} \, \|f\|_{\infty}}{\sqrt{N}}.
\end{equation}
\end{theorem}

\noindent\textbf{Proof:} See Appendix \ref{appThIS2}. \qed

The inequality in \eqref{eq:ThIS2} is well know for standard importance samplers \cite{agapiou2017importance}. It states that the approximation error is $\mathcal{O}(N^{-\frac{1}{2}})$. However, it also shows that, for given $N$, the error depends on the dimension $d_z$ of the nuisance variables through the constant $C_y^{d_z}$.  

In particular, let us assume without loss of generality that $\|g_y^{d_z}\|_\infty \le 1$ (and $\| l_y^{d_z} \|_\infty \le 1$ as a consequence). Then, the constant on the right-hand-side of \eqref{eq:ThIS2} can be shown to have the form (see Eq. \eqref{eq:C_ydz} in Appendix \ref{appThIS2}) 
\beq
C_y^{d_z} \leq \frac{\mathcal{C}_p}{\pi_0(l_y^{d_z}) } 
	,
\label{eqCy}
\eeq
where $\mathcal{C}_p$ is a finite constant (which only depends on $p$) that results from the application of the Marcinkiewicz-Zygmund inequality. 
Therefore, any dependence on the dimension of the nuisance variable $Z$ is encapsulated in the normalization constant $\pi_0(l_y^{d_z})$.

To analyze how the normalization constant~$\pi_0(l_y^{d_z})$ behaves as the dimension~$d_z$ grows, we begin with a simple analytically tractable model. In Section~\ref{sec:Linear_Gaussian} below, we examine a linear–Gaussian model in which the true posterior measure $\pi(\sd x)$ is Gaussian and can be computed exactly.


\section{Linear and Gaussian models} \label{sec:Linear_Gaussian}

\subsection{The model}\label{Gauss_Model} 

We consider a family of linear Gaussian models $\{\mM^{d_z}\}_{d_z \in \mbN}$, indexed by the dimension $d_z$
of the nuisance variables. To be specific, let $X \sim \mN(\sd x;\mu_x,\,\Sigma_x)$ denote a $d_x$-dimensional Gaussian vector with $d_x \times 1$ mean vector $\mu_x$ and $d_x \times d_x$ positive definite covariance matrix $\Sigma_x$, and let $U \sim \mN(\sd z; 0,Q^{d_z})$ and $V \sim \mN(\sd y; 0,R)$ be two zero-mean Gaussian vectors with covariance matrices $Q^{d_z} \in \Real^{d_z \times d_z}$ and $R \in \Real^{d_y\times d_y}$, respectively. $U$ and $V$ are mutually independent and also independent of $X$. Then, we construct the nuisance variable as
\beq
Z = H^{d_z}X + U, \label{eqExZ} 
\eeq
where $H^{d_z} \in \Real^{d_z \times d_x}$ is a fixed and known matrix, and the observation vector as 
\beq
Y = AX + B^{d_z}Z + V, \label{eqExY}
\eeq
where the matrices $A \in \Real^{d_y \times d_x}$ and $B^{d_z} \in \Real^{d_y \times d_z}$ are fixed and known as well.

Substituting \eqref{eqExZ} into \eqref{eqExY}, one can rewrite the observation vector $Y$ in terms of the state of interest $X$ and known model parameters alone, namely,
\beq
Y = T^{d_z} X + D^{d_z}, \quad \text{where} \quad T^{d_z}=A+B^{d_z}H^{d_z} \quad \text{and} \quad D^{d_z} = B^{d_z}U+ V.
\label{eqExYX}
\eeq
As a result, $Y$ is Gaussian with distribution $\mN(\sd y; \mu_y^{d_z},\Sigma_y^{d_z})$, where 
\beq
\mu_y^{d_z} = T^{d_z}\mu_x, 
\quad \text{and} \quad 
\Sigma_y^{d_z} = T^{d_z}\Sigma_x (T^{d_z})^\top + B^{d_z}Q^{d_z}(B^{d_z})^\top + R. \label{eq:mu_Sd}
\eeq
If we assume that  $T^{d_z}\Sigma_x(T^{d_z})^\top$, $B^{d_z}Q^{d_z}(B^{d_z})^\top$ and $R$ are positive definite, then 
\[
y^\top \Sigma_y^{d_z} y \ge 0,
\quad\text{for all }y \in \mathbb{R}^{d_y},
\]   
hence $\Sigma_y^{d_z}$ is also positive definite.

Let $\sff(y|x,z)$, $\sff(y|x)$ and $\sff(y)$ denote
\begin{itemize}
\item the conditional pdf of the r.v. $Y$ given $X=x$ and $Z=z$,
\item the conditional pdf of the r.v. $Y$ given $X=x$, and
\item the marginal pdf of the r.v. $Y$,
\end{itemize}
respectively. All three pdf's are Gaussian: $\sff(y|x,z)$ corresponds to Eq. \eqref{eqExY}, $\sff(y|x)$ corresponds to Eq. \eqref{eqExYX} and $\sff(y)$ is the density of the distribution $\mN(\sd y; \mu_y^{d_z},\Sigma_y^{d_z})$. For the observation model in \eqref{eqExY} the likelihood function can be written as 
\beq
g_y^{d_z}(x,z) = \exp\left\{
	-\frac{1}{2}\left(
		y - Ax - B^{d_z}z
	\right)^\top R^{-1} \left(
		y - Ax - B^{d_z}z
	\right)
\right\},   \notag
\eeq
which satisfies $\| g_y^{d_z} \|_\infty = 1$ for all $y$ and can be written as $g_y^{d_z}(x,z) = \sff(y|x,z)\sqrt{(2\pi)^{d_y}|R|}$, where $\abs{R}$ denotes the determinant of $R$. Then it readily follows that 
\beq 
l_y^{d_z}(x) = \int g_y^{d_z}(x,z)\kappa^{d_z}(x,\sd z) = \sff(y|x) \sqrt{(2\pi)^{d_y}|R|}
\nn
\eeq
and
\beq
\pi_0(l_y^{d_z}) = \sff(y) \sqrt{(2\pi)^{d_y}|R|} = \sqrt{
	\frac{|R|}{|\Sigma_y^{d_z}|} 
} \exp\left\{
	-\frac{1}{2}(y-\mu_y^{d_z})^\top(\Sigma_y^{d_z})^{-1}(y-\mu_y^{d_z})
\right\}.
\label{eq:pi0lyd}
\eeq
Note from Eq. \eqref{eq:pi0lyd} that the dependence of the integral $\pi_0(l_y^{d_z})$ on the dimension $d_z$ is encapsulated in the argument of the exponential term and in $|\Sigma_y^{d_z}|$, therefore, the accuracy of the nested IS approximation is governed by the structure of the model matrices in expression~\eqref{eq:mu_Sd}. In particular, if we can uniformly control both the exponential term and the determinant of $\Sigma_y^{d_z}$ w.r.t. the nuisance variable dimension $d_z$, then it is possible to build linear Gaussian models in which the approximation error remains uniformly bounded as $d_z$ grows. 
The previous discussion is formalized in the following theorem.

\begin{theorem}\label{thm:LG_pol}
Let \(\{\mM^{d_z}\}_{d_z \in \mathbb{N}}\) be the family of linear Gaussian models introduced above. Assume that there exists a polynomial \(\widetilde{P}_m : \mathbb{R} \mapsto \mathbb{R}\) of degree \(m \in \mathbb{N}\) such that
\beq\label{eq:SL_Pm}
\max \left\{ \sigma_1(B^{d_z})^{2d_y},\, \sigma_1(H^{d_z})^{2d_y}, \; \lambda_1(Q^{d_z})^{d_y} \right\}\leq \widetilde{P}_m(d_z), 
\quad \forall d_z \in \mbN.
\eeq
Then for any $r>0$, there exists a polynomial \(P_{n,r}(d_z)\) of degree at most \(n\leq 2m\), such that
\[
C_y^{d_z} \leq P_{n,r}(d_z), \quad \forall y\in B_{r}(\mu_y^{d_z}) := \{ y'\in\mbR^{d_y}: \| y' - \mu_y^{d_z} \| < r \},
\]
where $C_y^{d_z}$ is the constant in Theorem \ref{thIS2} (see also \eqref{eqCy}).  
Consequently, for any \(f \in B(\mX)\),
\[
\left\| \pi(f) - \pi^{N,M}(f) \right\|_p \leq \frac{P_{n,r}(d_z)\, \|f\|_\infty}{\sqrt{N}}, \quad \forall y\in B_{r}(\mu_y^{d_z}).
\]
\end{theorem}

\begin{proof}
See Appendix \ref{Ap:Linear_Gaussian}.
\end{proof}

\begin{remark}
Under the assumptions of Theorem~\ref{thm:LG_pol}, the number of samples required to guarantee a prescribed approximation error for the posterior expectation \(\pi(f)\) grows at most polynomially with the dimension $d_z$ of the nuisance variable $Z$, rather than exponentially. This highlights the role of structural properties (e.g., boundedness of singular values and eigenvalues) in mitigating the curse of dimensionality in IS.
\end{remark}

\begin{corollary}\label{cor:conv_unif_dz}
Under the assumptions of Theorem~\ref{thm:LG_pol}, if the inequality \eqref{eq:SL_Pm} holds with \(m = 0\), then there exists a constant \(C_y < \infty\), independent of \(N\), \(M\), \(f\), and \(d_z\), such that
\[
\left\| \pi(f) - \pi^{N,M}(f) \right\|_p \leq \frac{C_y\, \|f\|_\infty}{\sqrt{N}}, \quad \forall\, f \in B(\mX).
\]
In particular, $\lim_{N \to \infty} \| \pi(f) - \pi^{N,M}(f) \|_p = 0,$
uniformly in \(d_z\).
\end{corollary}

We emphasize that the key challenge in obtaining the previous results within the context of the family of linear Gaussian models $\mathcal{M}^{d_z}$ is the control of the exponential term and the determinant $|\Sigma_y^{d_z}|$ in Eq.~\eqref{eq:pi0lyd}. However, we have found that addressing the exponential term becomes more straightforward if we treat the observation $Y$ as a r.v. rather than a fixed value. This approach simplifies the analysis for this specific scenario, as discussed in Section \ref{ssec:Random observations}.


\subsection{Random observations}\label{ssec:Random observations} Let us assume now that, rather than arbitrary and fixed, the observation $Y$ is random and generated by the model described in Section \ref{Gauss_Model}. Within this setting, we can still obtain a positive lower bound for $\pi_0(l_Y^{d_z})$ and, hence, a finite upper bound for $C_Y^{d_z}$, which holds uniformly over $d_z\in \mbN.$ We use the subscript $Y$ instead of $y$ to emphasize the randomness of the observation. To be specific, let
 $\mathcal{G}=\sigma$-$(Y)$ denote the $\sigma$-algebra generated by the random observation $Y.$ Following exactly the same argument as in the proof of Theorem \ref{thIS2} one can show that 
\beq
\mbE \left[ \abs{\pi(f) - \pi^{N,M}(f)}^p \mid \mathcal{G} \right]^{\frac{1}{p}} \leq \frac{
	C_Y^{d_z} \| f \|_\infty
}{
\sqrt{N}
},  \notag
\eeq
 where  
 $$
 C_Y^{d_z} \leq \frac{\mathcal{C}_p \norm{g_Y^{d_z}}_\infty  }{\pi_0(l_Y^{d_z}) },
 $$
and $\mathcal{C}_p< \infty$ is a deterministic constant independent of $d_z$.  Note that, as a consequence, $C_Y^{d_z}$ and $\pi_0(l_Y^{d_z})$ are now r.v.'s, while $\| g_Y^{d_z}\|_{\infty}=1$ independently of $Y$ by construction. 
The integral $\pi_0(l_Y^{d_z})$ can be explicitly written as
\beq
\pi_0(l_Y^{d_z}) = \sqrt{
	\frac{|R|}{|\Sigma_Y^{d_z}|} 
} \exp\left\{
	-\frac{1}{2}(Y-\mu_y^{d_z})^\top\Sigma_Y^{-1}(Y-\mu_y^{d_z})
\right\}= \sqrt{
	\frac{|R|}{|\Sigma_Y^{d_z}|} 
} \exp \left(-\frac {\xi_Y^{2}}{2}\right),
\nn
\eeq
where the r.v. $\xi^2_Y:=(Y-\mu_y^{d_z})^\top\Sigma_Y^{-1}(Y-\mu_y^{d_z})$ is distributed according to a $\chi^2$ probability law with  $d_y$ degrees of freedom (where $d_y$ is the dimension of $Y$) \cite{koch2007introduction} and, therefore, it is independent of $d_z$. 
The determinant of $\Sigma_Y^{d_z}$ can be bounded in the same manner as shown in Lemma~\ref{lem:detT} in Appendix~\ref{Ap:Linear_Gaussian}. As a result, we can obtain an upper bound for the error that remains uniform in the dimension $d_z$ even in this random setting. 

As highlighted in this section, exploiting the randomness and structure of the observation $Y$ can significantly simplify certain aspects of the error bound analysis within the IS scheme. We explore this in further detail in Section \ref{sec:General_Random} below.


\section[Models with random observations]{General models with random observations}\label{sec:General_Random}

In the preceding sections, we have examined the approximation error \( \| \pi(f) - \pi^{N,M}(f) \|_p \) for a fixed observation $Y=y$. In this section, we adopt the viewpoint that $Y$  itself is a r.v. With this perspective, the likelihood function and the posterior measure also become random objects, in particular, the quantity $\mbE\left[ | \pi(f) - \pi^{N,M}(f) |^p \vert \mG \right]^{\frac{1}{p}}$ becomes random as well, because it depends on $Y$. To derive an upper bound on this random error, we take the expectation w.r.t. $Y$. This leads us naturally to examine how integration over the random observation is performed. The central tool for this development is the Bochner integrability \cite{dashti2015bayesian} of a specially constructed function, which acts as a link between the fixed and random observation formulations.

\subsection{Random posterior measure}
As hinted in Section \ref{ssec:Random observations}, the approximation error \( \| \pi(f) - \pi^{N,M}(f) \|_p \) can become easier to be analyzed when adopting the perspective that the observation $Y$ is a r.v. generated by the model in Section \ref{ssec:sIS2}.    
In that setting both
\begin{equation*}
    l_Y^{d_z}(x) := \int_{\mathcal{Z}} g_Y^{d_z}(x, z) \, \kappa^{d_z}(x, \sd z), 
\end{equation*}
and \( \pi_0(l_Y^{d_z}) \) are real-valued r.v.'s. 
Consequently, the posterior distribution becomes a random probability measure, which we denote as \( \pi_Y(\sd x) \propto l_Y^{d_z}\pi_0(\sd x) \) with the subscript $_Y$ included to emphasize the randomness of the observation.

\medskip
Throughout this section, we assume that \( g_y^{d_z} \in L^2(\mix^{d_z}) \) for all \( y \in \mathcal{Y} \), where the joint measure \( \mix^{d_z} \) is defined in Eq.~\eqref{eq:Mix_meas}. We further assume that \( g_y^{d_z}(x, z) > 0 \) for all \( (x, z) \in \mathcal{X} \times \mathcal{Z} \) and all \( y \in \mathcal{Y} \). In addition, \( g_y^{d_z} \) is now assumed to be normalized, namely,
\begin{equation*} 
    \int_{\mathcal{Y}} g_y^{d_z}(x, z) \, \lambda(\sd y) = 1, \quad \forall (x, z) \in \mathbb{R}^{d_x \times d_z},
\end{equation*}
where \( \lambda \) denotes the Lebesgue measure on \( \mathcal{Y} \). Therefore \( g_y^{d_z} \) coincides with the conditional pdf of \( Y \) given \( (X, Z) = (x, z) \). This normalization assumption is introduced to facilitate the analysis. However, it does not affect the construction of the estimators in Algorithm~\ref{A2}, which rely on normalized weights. In other words, the estimators \( l_Y^M \) and \( \pi_Y^{N,M} \) can still be constructed using a likelihood function of the form $g_y^{d_z}(x,z)= cg_y^{d_z}(y \mid x,z)$ for an arbitrary constant $c>0$ independent of $x$ and $z$.

Much of the analysis to be elaborated in this section depends on two objects that we introduce next:
\begin{itemize}
\item The probability measure \( \eta \in \mP(\mY) \), defined as
\[
\eta(\sd y) :=  \pi_0(l_y^{d_z}) \, \lambda(\sd y),
\]
corresponds to the marginal probability law of the r.v. $Y$ induced by the model $\mM^{d_z}$. 

\item The measurable mapping
\beq
\begin{array}{llll}
    \ell: &\mathcal{Y} &\mapsto &L^2(\mix^{d_z})\\
    &y &\leadsto &\ell_y := \frac{g_y^{d_z}(\cdot,\cdot)}{\mix^{d_z}(g_y^{d_z})},
\end{array}
\label{eq:link}
\eeq
which we refer to as the \emph{link function}, assigns to each observation \( y \) the corresponding normalized density over \( \mX\times \mZ \). Since \( \mix^{d_z}(g_y^{d_z}) = \pi_0(l_y^{d_z}) > 0 \), the map \( \ell_y \) is well defined.
\end{itemize}

For any test function \(\varphi \in L^2(\pi_0)\), define its constant extension
\begin{equation*}
\widetilde\varphi(x,z) := \varphi(x).
\end{equation*}
Trivially, $\widetilde \varphi \in L^2(\mix^{d_z})$.  The posterior expectation under the random observation \( Y \) satisfies
\begin{equation}\label{eq:pi_Y_Inner}
    \pi_Y(\varphi) = \int_\mathcal{X} \varphi(x) \, \pi_Y(\sd x) 
= \frac{\int_\mathcal{X} \varphi(x)\, l_Y^{d_z}(x) \, \pi_0(\sd x)}{\pi_0(l_Y^{d_z})} 
= \left\langle \widetilde\varphi, \ell_Y \right\rangle_{L^2(\mix^{d_z})},
\end{equation}
where $\langle \cdot, \cdot \rangle_{L^2(\mix^{d_z})}$ denotes the inner product defined in the Hilbert space $L^2(\mix^{d_z})$, i.e. 

$$ \left\langle \widetilde\varphi, \ell_Y \right\rangle_{L^2(\mix^{d_z})}= \int_\mX \int_\mZ \widetilde \varphi(x,z) \ell_Y(x,z) \mix^{d_z}( \sd (x, z)).$$

\subsection{Bochner integrability}\label{sec:Bochner_Int}
Take test functions $\varphi,\psi \in L^2(\pi_0)$. The quantity \( \mathbb{E}[\pi_Y(\varphi)] \) defines a linear functional on \( L^2(\pi_0) \); meanwhile, the quantity \( \mathbb{E}[\pi_Y^2(\varphi)] \) defines a quadratic form on \( L^2(\pi_0) \) induced by the symmetric bilinear form 
\beq
\begin{array}{lccl}
B : &L^2(\pi_0) \times L^2(\pi_0) &\mapsto &\mbR\\
&(\varphi, \psi) &\leadsto &\mathbb{E}[\pi_Y(\varphi)\, \pi_Y(\psi)].    
\end{array}
\nn
\eeq
By the inner-product representation of the random measure $\pi_Y$  in Eq.~\eqref{eq:pi_Y_Inner}, the boundedness of  \( \mathbb{E}[\pi_Y^p(\varphi)] \) with $p\in\{1,2\}$, is closely related to the Bochner integrability of the map \( y \mapsto \ell_y \). 

The function \( \ell \) is called Bochner $p$-integrable (see Proposition~1.2.2 in \cite{tuomas2016analysis}) if, and only if,
\begin{enumerate}[label=\roman*)]
    \item \( \ell \) is strongly \( \eta \)-measurable (see Definition~1.1.14 in \cite{tuomas2016analysis}), and
    \item the following integral is finite
  \begin{equation}
    \label{eq:Kpdz}
    K_p^{d_z} := \int_{\mathcal{Y}} \| \ell_y \|_{L^2(\mix^{d_z})}^p \, \eta(\sd y) < \infty.
  \end{equation}
\end{enumerate}

The first condition is satisfied under mild regularity assumptions on \( g_y^{d_z} \), which already holds for the models of interest (see Lemma~\ref{lem:S_M}). Therefore, for the remainder of the paper, we only require the finiteness of the integral $K_p^{d_z} $ in \eqref{eq:Kpdz} to ensure Bochner integrability.

For \( 1 \leq p \leq \infty \), we denote by \( L^p(\mathcal{Y}; L^2(\mix^{d_z})) \) the Banach space of Bochner \( p \)-integrable functions \( \upsilon : \mathcal{Y} \mapsto L^2(\mix^{d_z}) \) w.r.t. the measure \( \eta \) (cf. \cite{tuomas2016analysis}). In particular, for \( p = \infty \), the norm is defined as
\[
\|\upsilon\|_{L^\infty(\mathcal{Y}; L^2(\mix^{d_z}))} := \operatorname*{ess\,sup}_{y \in \mathcal{Y}} \| \upsilon_y \|_{L^2(\mix^{d_z})} < \infty.
\]
Since \( \eta \) is a probability measure, for $1\leq p < \infty$, we have the continuous embedding 
\[
L^{\infty}(\mathcal{Y}; L^2(\mix^{d_z})) \subseteq L^{p+1}(\mathcal{Y}; L^2(\mix^{d_z})) \subseteq L^p(\mathcal{Y}; L^2(\mix^{d_z})).
\]

With this notation, we formalize the above discussion in the following result.
\begin{theorem}\label{thm:4.1}
Assume that \( \ell \in  L^p(\mathcal{Y}; L^2(\mix^{d_z})) \) for \( p \in \{1,2\} \). Then, for all test functions \( \varphi \in L^2(\pi_0) \),
\[
\mathbb{E}[\pi_Y^p(\varphi)] \leq K_p^{d_z} \, \|\varphi\|_{L^2(\pi_0)}^p,
\]
where $K_p^{d_z}<\infty$ is the constant in \eqref{eq:Kpdz} and $p\in\{1,2\}$.
\end{theorem}

\begin{proof}
See Lemma \ref{lem:bounded_functional} and Remark \ref{rem:proof_cor} in Appendix \ref{apx:Bochner-F}.
\end{proof}


\subsection{Approximation errors}

We now provide a convergence result for the approximation of the posterior expectation with random observations.

\begin{theorem}\label{thm:thIS2R}
For all $y\in \mY$, let \( g_y^{d_z}(x,z)\in L^2(\mix^{d_z}) \) be a strictly positive conditional pdf. Then, for any \( f\in B(\mathcal{X})\) the following statements are equivalent:
\begin{enumerate}[label=\roman*)]
    \item There is a constant \( C_2^{d_z} < \infty \), independent of \( N \), \( M \), and \( f \), such that
    \begin{equation}\label{eq:Iff_error}
        \left\| \pi_Y(f) - \pi_Y^{N,M}(f) \right\|_2 \leq \frac{C_2^{d_z} \| f \|_\infty}{\sqrt{N}}.
    \end{equation}
    \item  \( \ell \in L^2\big( \mathcal{Y}; L^2(\mix^{d_z}) \big) \).
\end{enumerate}
Moreover 
\[
C_2^{d_z} \leq B_2 K_2^{d_z},
\]
where \( B_2 \) is independent of $d_z$, and $K_2^{d_z}$ is defined in \eqref{eq:Kpdz}.
\end{theorem}
\begin{proof}
See Appendix~\ref{appThIS3}.
\end{proof}

Theorem \ref{thm:thIS2R} provides a sufficient and necessary condition for the $L^2$ approximation errors of nested IS Algorithm \ref{A2} to converge with rate $\mathcal{O}(N^{-1/2})$. It also provides the means to establish a control on the effect of the dimension $d_z$ on the approximation error bounds. This is specifically provided by Theorem \ref{thIS2R} below.

\begin{theorem}\label{thIS2R}
For all $y\in \mY$, let \( g_y^{d_z}(x,z) \in L^2(\mix^{d_z})\) be strictly positive conditional pdf. Assume that \(\ell \in  L^p\big( \mathcal{Y}; L^2(\mix^{d_z}) \big) \) for some \( p \in \{1, 2\} \),  and there exists a polynomial \( \widetilde P_{n,p}(d_z) \) of degree \(m\), possibly depending on \(p\), such that
\begin{equation*}
K_p^{d_z} \leq \widetilde P_{m,p}(d_z), \quad \forall d_z \in \mbN, 
\end{equation*} where $K_p^{d_z}$ is defined in \eqref{eq:Kpdz}.
 Then, for any \( f \in B(\mX) \), there exists a polynomial \(P_{n,p}(d_z)\) of degree at most \(n\leq m\), possibly depending on \(p\), such that
\[
\left\| \pi_Y(f) - \pi_Y^{N,M}(f) \right\|_p 
\leq \frac{P_{n,p}(d_z)\, \|f\|_\infty}{\sqrt{N}}
\]
for the assumed value of $p\in\{1,2\}$.
\end{theorem}

\begin{proof}
See Appendix~\ref{ap:T4.3}.
\end{proof}

\begin{remark}
Under the assumptions of Theorem~\ref{thIS2R}, the number of samples needed to accurately approximate the posterior expectation \(\pi_Y(f)\) increases at most polynomially with the dimension \(d_z\) of the nuisance variables, rather than exponentially. This illustrates how structural properties of the model can effectively alleviate the curse of dimensionality in the nested importance sampler.
\end{remark}

\begin{remark}
The norm in Theorem~\ref{thm:thIS2R} is not the conditional expectation $$\mbE\left[ \left| \pi_Y(f) - \pi_Y^{N,M}(f) \right|^p \vert \mG \right]^{\frac{1}{p}},$$ but rather reflects the full randomness in the estimator \( \pi_Y^{N,M} \) as well as in the observation \( Y \).
\end{remark}

\begin{corollary}\label{cor:IS_uniform}
Under the assumptions of Theorem~\ref{thIS2R}, if \( m = 0 \), i.e., if there is a constant \( \mathcal{K}_p < \infty \) independent of $d_z$ such that
\begin{equation}\label{eq:U_Kdz}
\sup_{d_z \in \mathbb{N}} K_p^{d_z} \leq \mathcal{K}_p,
\end{equation}
where $K_p^{d_z}$ is defined in \eqref{eq:Kpdz},
then, for any \( f \in B(\mX) \), there exists a constant \( C_p < \infty \), independent of \( N \), \( M \), \( f \), and the dimension \( d_z \), such that
\[
\left\| \pi_Y(f) - \pi_Y^{N,M}(f) \right\|_p 
\leq \frac{C_p \, \|f\|_\infty}{\sqrt{N}}, \quad \text{for $p\in\{1,2\}$.}
\]
In particular, $\lim_{N \to \infty} \| \pi_Y(f) - \pi_Y^{N,M}(f) \|_p = 0,$
uniformly over \(d_z\).
\end{corollary}

\begin{remark}\label{rem:C_p_BK}
The constant \( C_p < \infty\) appearing in Corollary~\ref{cor:IS_uniform} can be decomposed as
\[
C_p = B_p \mathcal{K}_p,
\]
where \( B_p \) is a constant depending only on \( p \), arising from the Marcinkiewicz–Zygmund inequality, and \( \mathcal{K}_p \) captures the bounds related to the Bochner integrability of the link function $\ell$.
\end{remark}


\section{Special cases}\label{sec:Special_cases}

In this section, we explore specific scenarios under which the assumptions of Corollary~\ref{cor:IS_uniform} are satisfied, i.e., where inequality \eqref{eq:U_Kdz} holds. We concentrate on the case \(p = 2\), which supersedes \(p = 1\). First note that, using \eqref{eq:link} and the identity \( \eta(\sd y) = \mix^{d_z}(g_y^{d_z}) \, \lambda(\sd y) \), we obtain
\beq
\label{eq:C_Fb} 
\int_{\mY} \| \ell_y \|_{L^2(\mix^{d_z})}^p \, \eta(\sd y) =\int_{\mY} \frac{ \| g_y^{d_z} \|_{L^2(\mix^{d_z})}^p }{ \mix^{d_z}(g_y^{d_z})^{p-1} } \, \lambda(\sd y).
\eeq
Throughout this section, we omit the superscript \( d_z \) in the likelihood function for notational simplicity (we write \( g_y \) instead of \( g_y^{d_z} \)). Nevertheless, it should be understood that all models under consideration depend on the dimension \( d_z \).

The following remark plays a key role in our analysis.
\begin{remark}\label{rem:link_norm}
Observe that \[
\| \ell_y \|_{L^2(\mix^{d_z})}^2 = \frac{\mix^{d_z}(g_y^2)}{\mix^{d_z}(g_y)^2},
\] is directly related to the chi-square divergence between the posterior and prior distributions ($\pi$ and $\pi_0$). This factor has been exploited for the analysis of importance samplers in earlier work, see, e.g., \cite{agapiou2017importance,akyildiz2021convergence}. Integrating $\| \ell_y \|_{L^2(\mix^{d_z})}^2$  over the observation space yields
\[
\int_{\mY} \| \ell_y \|_{L^2(\mix^{d_z})}^2 \, \eta(\sd y) 
= \int_{\mX \times \mZ} \left( \int_{\mY} \frac{g_y^2(x,z)}{\mix^{d_z}(g_y)} \, \lambda(\sd y) \right) \mix^{d_z}(\sd (x,z) ).
\]
Finally, if \( \ell_Y \) is regarded as a r.v., the Bochner $p$-integrability condition in~\eqref{eq:Kpdz} is equivalent to the moment condition
\begin{equation}\label{eq:BI_as_Ex}
\mathbb{E}\left[\| \ell_Y \|_{L^2(\mix^{d_z})}^p\right] < \infty.
\end{equation}
\end{remark}

%
\subsection{Linear Gaussian model}\label{sub:LG_PC}

We now show that, in the linear and Gaussian model, the integral \( K_2^{d_z} \) is in fact bounded by a finite constant \( \mathcal{K}_2 \) independent of the dimension \( d_z \). This implies that the assumptions of Corollary~\ref{cor:IS_uniform} are satisfied in this setting.

Recall the family of linear Gaussian models $\{\mM^{d_z}\}_{d_z \in \mbN}$, indexed by the dimension $d_z$ of the nuisance variables presented in Section \ref{Gauss_Model}. The likelihood function for this model is Gaussian and can be written as 
\[
g_y(x,z) = \mathcal{N}( \sd y;Ax + B^{d_z}z, R),
\]
which, after taking squares, yields
\[
g_y^2(x,z) = \frac{1}{(2\pi)^{d_y} |R|} \exp\left\{
	-\frac{1}{2} \left( y - Ax - B^{d_z}z \right)^\top (2R^{-1}) \left( y - Ax - B^{d_z}z \right)
\right\}.
\]
It is easy to see that $g_y^2$ is proportional to a Gaussian density with halved covariance, i.e.,
\[
g_y^2(x,z) \propto \mathcal{N} \left( \sd y; Ax + B^{d_z}z, \frac{1}{2}R\right).
\]
Therefore,
\[
\int g_y^2(x,z)\, \mix^{d_z}(\sd( x,z)) = \frac{1}{(2\pi)^{d_y} |R|}  \mathcal{N}( \sd y;\mu_y^{d_z}, \Sigma_y^{d_z} - \tfrac{1}{2}R),
\]
where $\mu_y^{d_z}$ and $\Sigma_y^{d_z}$ are as defined in \eqref{eq:mu_Sd}. From \eqref{eq:C_Fb} and \eqref{eq:BI_as_Ex}, we observe that in order to satisfy the integrability condition \eqref{eq:Kpdz} it suffices to find a finite bound for
\begin{equation}\label{eq:norm_ell_Gauss}
\mathbb{E}\left[\| \ell_Y \|_{L^2(\mix^{d_z})}^2\right]
= \frac{1}{(2\pi)^{d_y} |R|} \int \frac{ \mathcal{N}( \sd y;\mu_y^{d_z}, \Sigma_y^{d_z} - \tfrac{1}{2}R) }{ \mathcal{N}( \sd y;\mu_y^{d_z}, \Sigma_y^{d_z}) } \, \lambda(\sd y).
\end{equation}

Let us denote $S_1 := \Sigma_y^{d_z}$ and $S_2 := \Sigma_y^{d_z} - \tfrac{1}{2}R$, so that the integrand in \eqref{eq:norm_ell_Gauss} becomes
\[
 \frac{ \mathcal{N}( \sd y;\mu_y^{d_z}, S_2) }{ \mathcal{N}( \sd y;\mu_y^{d_z}, S_1) }
= \left(\frac{ |S_1| }{ |S_2| } \right)^{\frac{1}{2}} \exp\left\{ -\frac{1}{2} (y - \mu_y^{d_z})^\top (S_2^{-1} - S_1^{-1}) (y - \mu_y^{d_z}) \right\}.
\]
Note that both \( S_1 \) and \( S_2 \) are positive definite, and
\[
S_1 = \Sigma_y^{d_z} \succ \Sigma_y^{d_z} - \tfrac{1}{2}R = S_2,
\]
where \( \succ \) denotes the Loewner order (see Section 7.7 in ~\cite{horn2012matrix}). Let us define
\beq 
\label{eq_SS2S1}
S^{-1} := S_2^{-1} - S_1^{-1}.
\eeq
By the operator monotonicity of matrix inversion on the cone of positive definite matrices (see Theorem 24 in \cite{magnus2019matrix}), it follows that
\[
S_1 \succ S_2 \quad \Rightarrow \quad S_1^{-1} \prec S_2^{-1} \quad \Rightarrow \quad S^{-1} \succ 0
\]
(hence, $S^{-1}$ is positive definite). The integrand in \eqref{eq:norm_ell_Gauss} is thus proportional to the density of a Gaussian distribution with mean $\mu_y^{d_z}$ and covariance $S$. As a consequence, the integral in \eqref{eq:norm_ell_Gauss} can be computed exactly to yield
\begin{equation}\label{eq:int_g2_g1_Gauss}
 \int \frac{ \mathcal{N}( \sd y;\mu_y^{d_z}, S_2) }{ \mathcal{N}( \sd y;\mu_y^{d_z}, S_1) } \, \lambda(\sd y)
= \left( \frac{1}{(2\pi)^{d_y} |S|}  \frac{ |S_1| }{ |S_2| } \right)^{\frac{1}{2}}.
\end{equation}

From \eqref{eq_SS2S1}, the identity $S_2^{-1}-S_1^{-1} = S_1^{-1}(S_1-S_2)S_2^{-1}$, and since $S_1-S_2=\cfrac{1}{2}R$, we have 
\beq
S^{-1} = S_1^{-1}\left(\cfrac{1}{2}R\right)S_2^{-1}, 
\nn
\eeq
which readily yields
\beq
S = S_2(2R^{-1})S_1
\label{eq__a}
\eeq
and 
\beq
\abs{S} = 2 \frac{\abs{S_1}\abs{S_2}}{\abs{R}}.
\label{eq__b}
\eeq
Substituting \eqref{eq__a} and \eqref{eq__b} into \eqref{eq:int_g2_g1_Gauss} we arrive at
\beq
\int \frac{ \mathcal{N}(\sd y;\mu_y^{d_z}, S_2) }{ \mathcal{N}(\sd y; \mu_y^{d_z}, S_1) } \, \lambda(\sd y) = 
\frac{|R|^{\frac{1}{2}}}{2^{\frac{1}{2}} (2\pi)^{\frac{d_y}{2}} |S_2|}.
\label{eq__c}
\eeq
Combining \eqref{eq__c} and \eqref{eq:norm_ell_Gauss} we readily see that
\beq
\label{eq:B_LG}
\mathbb{E}\left[\| \ell_Y \|_{L^2(\mix^{d_z})}^2\right]
= \frac{1}{2^{\frac{1}{2}} (2\pi)^{\frac{3}{2}d_y} |R|^{\frac{1}{2}}} \frac{1}{|S_2|}=:K_2^{d_z},
\eeq
which is always finite, ensuring that the moment condition \eqref{eq:BI_as_Ex} holds (whenever $R \succ 0$, as assumed for the family of models) and so does the Bochner integrability condition \eqref{eq:Kpdz}. Moreover, the quantity in \eqref{eq:B_LG} can be bounded in terms of the model parameters (cf.~Lemma~\ref{lem:detT}). Indeed, recalling \eqref{eq:mu_Sd}, we readily see that
\[
S_2 = \Sigma_y^{d_z} - \tfrac{1}{2} R
= T^{d_z} \Sigma_x (T^{d_z})^\top + B^{d_z} Q^{d_z} (B^{d_z})^\top + \tfrac{1}{2} R.
\]
If we additionally observe that $S-\frac{1}{2}R\succ 0$ then, by Theorem 25 in \cite{magnus2019matrix},
\[
|S_2|\geq \cfrac{1}{2^{d_y}}|R|,
\quad
\text{hence}
\quad
\frac{1}{|S_2|} \leq \frac{2^{d_y}}{|R|},
\] 
and the expectation in \eqref{eq:B_LG} can be bounded as
\begin{equation*}
\mathbb{E}\left[\| \ell_Y \|_{L^2(\mix^{d_z})}^2\right] = K_2^{d_z} 
\leq \frac{2^{d_y-\frac{1}{2} }}{ \left[ (2\pi)^{d_y} |R| \right]^{\frac{3}{2}}}=:\mathcal{K}_2 < \infty,
\end{equation*}
where \(\mathcal{K}_2\) is, by construction, independent of \(d_z\) (see Subsection \ref{Gauss_Model}).

\begin{remark}
In the linear and Gaussian case, the Bochner 2-integrability condition given by Eq. \eqref{eq:B_LG} enables us to find a finite constant \( \mathcal{K}_2<\infty \), independent of \( d_z \), that satisfies the assumptions of Corollary~\ref{cor:IS_uniform} and, therefore, \eqref{eq:B_LG} is sufficient to prove that the error bounds for the nested importance sampler (Algorithm \ref{A2}) hold uniformly over $d_z$ and converge towards 0 as $\mO(N^{-\frac{1}{2}})$.
\end{remark}

\subsection[Uniformly bounded observation function]{Uniformly bounded observation function}\label{sub:BfO_GN}

Let us consider the observation model
\begin{equation}
y = f(x, z) + \epsilon,
\label{eqBoundedSensor}
\end{equation}
where $\epsilon \sim \mathcal{N}( \sd y; 0, R)$ is a Gaussian r.v. independent of \( X \) and \( Z \). The function \( f : \mathcal{X} \times \mathcal{Z} \mapsto \mathbb{R}^{d_y} \) is assumed to be deterministic and uniformly bounded in the Euclidean norm, i.e., 
\beq
\sup_{(x,z)\in\mathcal{X} \times \mathcal{Z}} \| f(x,z) \|_2 \leq F,
\label{eq__F}
\eeq
for some constant \( F < \infty \) independent of the nuisance dimension \( d_z \).
The noise covariance matrix $R \in \mathbb{R}^{d_y \times d_y}$ is assumed to be symmetric and positive definite.

We define the Mahalanobis inner product and its induced norm w.r.t. $R$ for $y_1, y_2 \in \mathbb{R}^{d_y}$ as
\begin{equation*}
\langle y_1, y_2 \rangle_R := y_1^\top R y_2, \quad
\text{and}
\quad
\|y_2\|_R := \sqrt{y_2^\top R y_2},
\end{equation*}
respectively. In addition, let \( R^{\frac{1}{2}} \) denote the unique symmetric positive definite square root of \( R \). Then, the \( R \)-norm of any vector \( y \) satisfies \( \|y\|_R = \|R^{\frac{1}{2}} y\|_2 \). Let \( \lambda_1(R) \) be the largest eigenvalue of \( R \); the induced operator 2-norm of \( R^{\frac{1}{2}} \) is given by \( \|R^{\frac{1}{2}}\|_2 = \sqrt{\lambda_1(R)} \). Using the standard inequality \( \|A v\|_2 \le \|A\|_2 \|v\|_2 \), it follows that
\[
\|f(x,z)\|_R = \|R^{\frac{1}{2}} f(x,z)\|_2 \le \|R^{\frac{1}{2}}\|_2 \, \|f(x,z)\|_2 \le \sqrt{\lambda_1(R)} F,
\]
where we have used \eqref{eq__F} for the last inequality. 

Let us additionally define
\[
F_R := \sup_{(x,z)} \|f(x,z)\|_R,
\]
which satisfies \( F_R \le \sqrt{\lambda_1(R)} F \). Applying the triangle inequality, and its reverse for the Mahalanobis norm, to the difference \( y - f(x, z) \) and then using the bound \( \|f(x, z)\|_R \le F_R \), we obtain (upon squaring each side) the inequality  
\begin{equation}\label{eq:Sq_ineq}
\big( \|y\|_R - \|f(x,z)\|_R \big)^2 \le \|y - f(x,z)\|_R^2 \le \big( \|y\|_R + F_R \big)^2.
\end{equation}
Now, note that the conditional pdf of $y$ given $x, z$ can be written as
\[
g_y(x,z) = \frac{1}{(2\pi)^{\frac{d_y}{2}} \abs{R}^{\frac{1}{2}}} \exp\left\{ -\tfrac{1}{2} \|y - f(x,z)\|_R^2 \right\}
\]
and define 
$$
C_R := \left( (2\uppi)^{d_y} \abs{R} \right)^{-1/2}
$$ 
for conciseness. Using \eqref{eq:Sq_ineq} we readily obtain the bounds
\beq\label{ineq:Bounds_gy}
C_R \exp\left\{ -\tfrac{1}{2} (\|y\|_R + F_R)^2 \right\} \leq g_y(x,z) \leq C_R \exp\left\{ -\tfrac{1}{2} \big(\|y\|_R - \|f(x,z)\|_R\big)^2 \right\}.
\eeq
The lower bound in \eqref{ineq:Bounds_gy} is independent of $(x,z)$, hence we have 
\beq\label{ineq:gy}
C_R \exp\left\{ -\tfrac{1}{2} (\|y\|_R + F_R)^2 \right\} \leq \mix^{d_z}(g_y),
\eeq
while from the right-hand-side of \eqref{ineq:Bounds_gy} we have
\beq\label{ineq:gy2}
g_y^2(x,z) \leq C_R^2 \exp\left\{ - \big(\|y\|_R - \|f(x,z)\|_R\big)^2 \right\}.
\eeq
Combining \eqref{ineq:gy} and \eqref{ineq:gy2} we arrive at
\beq\label{ratio_bound}
\frac{g_y^2(x,z)}{\mix^{d_z}(g_y)} \leq C_R \exp\left\{ \tfrac{1}{2} (\|y\|_R + F_R)^2 - (\|y\|_R - \|f(x,z)\|_R)^2 \right\}.
\eeq
Expanding the exponent and completing the square on the right-hand-side of \eqref{ratio_bound}, it is straightforward to show that
\beq\label{B_Ratio}
\int_\mY \frac{g_y^2(x,z)}{\mix^{d_z}(g_y)} \, \lambda(\sd y)\leq C_R \exp\left\{ \tfrac{11}{4} F_R^2 \right\} \int_\mY
\exp\left\{ -\tfrac{1}{2} \left( \|y\|_R - \tfrac{3}{2} F_R \right)^2 \right\}\,  \lambda(\sd y).
\eeq
From Remark \ref{rem:link_norm}, we see that Bochner integrability is guaranteed if the dominating function $y \mapsto \exp\left\{ -\tfrac{1}{2} (\|y\|_R - \tfrac{3}{2}F_R)^2 \right\}$ is integrable, which we show next.

\subsubsection*{Integrability of the dominating function}
Let \( c = \tfrac{3}{2} F_R \).  
We now verify the integrability of the function
\beq\label{eq:h(y)}
h(y) := \exp\left( -\tfrac{1}{2} \left\{ \|y\|_R - c \right)^2 \right\}.
\eeq
Performing the change of variable \(\widetilde{y} = R^{\frac{1}{2}} y\), whose Jacobian determinant is \(\lvert R \rvert^{\frac{1}{2}}\), the volume element transforms as $\sd y = \lvert R \rvert^{-\frac{1}{2}} \, \sd \widetilde{y}.$
Hence,
\beq\label{eq:h(y)lam(dy)}
\int_{\mathcal{Y}} h(y) \, \lambda(\sd y) 
= \lvert R \rvert^{-1/2} \int_{\mathcal{Y}} \exp\left\{ -\tfrac{1}{2} \left( \|\widetilde{y}\|_2 - c \right)^2 \right\} \lambda(\sd \widetilde{y}).
\eeq

Next, we transform the integral from Cartesian coordinates to \( d_y \)-dimensional hyperspherical coordinates \((r, \phi_1, \ldots, \phi_{d_y-1})\), where $r=\|\widetilde{y}\|_2.$
In these coordinates, the volume element transforms as
\[
\sd \widetilde{y} = r^{d_y - 1} \, \sd r \, \sd \omega_{d_y - 1},
\]
where $\sd \omega_{d_y-1}$ is the differential solid angle in $d_y-1$ dimensions.
The integral in \eqref{eq:h(y)lam(dy)} becomes
\beq\label{eq:h(y)lam(dy)_Sph}
\int_{\mY} h(y) \lambda(\sd y)=|R|^{-\frac{1}{2}} \int_0^\infty \int_{\text{full range}} \exp\left\{ -\tfrac{1}{2} (r - c)^2 \right\} r^{d_y - 1} \, \sd r \, \sd \omega_{d_y - 1}.
\eeq
Note that the integral above separates into radial and angular parts. The angular integral yields the surface area of the unit \((d_y - 1)\)-sphere,
\[
\mS_{d_y - 1} := \int_{\text{full range}} \sd \omega_{d_y - 1} = \frac{2 \uppi^{\frac{d_y}{2}}}{\Gamma\left( \frac{d_y}{2} \right)},  
\]
where $\Gamma(\cdot)$ is the Gamma function. Thus, we can take \eqref{eq:h(y)lam(dy)_Sph} and obtain the simplification
\beq\label{eq:Yellow_Box}
\int_{\mY}h(y)\lambda(\sd y) = \frac{\mS_{d_y - 1}}{|R|^{\frac{1}{2}}} \int_0^\infty \exp\left\{ -\tfrac{1}{2} (r - c)^2 \right\} r^{d_y - 1} \, \sd r.
\eeq
Let us denote the radial integral as
\beq\label{eq:I}
\mI := \int_0^\infty \exp\left\{ -\tfrac{1}{2} (r - c)^2 \right\} r^{d_y - 1} \, \sd r.
\eeq
Since the integrand is non-negative for \( r \geq 0 \), extending the integration range to \((-\infty, \infty)\) and replacing \( r^{d_y - 1} \) by \( |r|^{d_y - 1} \) 
and performing the change of variable \( u = r - c \)  (so that \( r = u + c \) and \( \sd r = \sd u \) ) yields the upper bound 
\[
\mI \leq \int_{-\infty}^\infty \exp\left\{ -\tfrac{1}{2} u^2 \right\} |u + c|^{d_y - 1} \, \sd u.
\]
Using the elementary inequality $|a + b|^k \leq 2^{k - 1} \big( |a|^k + |b|^k \big),$
we can therefore bound the integral \(\mI\) as
\begin{align}
\mI &\leq \int_{-\infty}^\infty \exp\left\{ -\tfrac{1}{2} u^2 \right\} 2^{d_y - 2} \big( |u|^{d_y - 1} + |c|^{d_y - 1} \big) \, \sd u \notag \\
&= 2^{d_y - 2} \left( \int_{-\infty}^\infty \exp\left\{ -\tfrac{1}{2} u^2 \right\} |u|^{d_y - 1} \, \sd u + |c|^{d_y - 1} \int_{-\infty}^\infty \exp\left\{ -\tfrac{1}{2} u^2 \right\} \, \sd u \right). \label{ineq:Up_Bound_I}
\end{align}
Both integrals in \eqref{ineq:Up_Bound_I} are standard and well-known. Specifically, the second term yields
\beq\label{eq:Gaus_NC}
\int_{-\infty}^\infty \exp\left\{ -\tfrac{1}{2} u^2 \right\} \, \sd u = \sqrt{2\uppi},
\eeq
while the first integral can be expressed in terms of the Gamma function,
\begin{align}\label{eq:first_Int}
\int_{-\infty}^\infty \exp\left\{ -\tfrac{1}{2} u^2 \right\} |u|^{d_y-1} \sd u = 2 \int_0^\infty \exp\left\{ -\tfrac{1}{2} u^2 \right\} u^{d_y-1} \sd u
= 2^{d_y/2} \Gamma\left( \frac{d_y}{2} \right). \end{align}
Substituting \eqref{eq:Gaus_NC} and \eqref{eq:first_Int} back into \eqref{ineq:Up_Bound_I}
and replacing $c = \tfrac{3}{2}F_R$, we obtain 
\beq\label{ineq:Final_I}
 \mI \le 2^{d_y-2} \left( 2^{\frac{d_y}{2}} \Gamma\left( \frac{d_y}{2} \right) + \left(\tfrac{3}{2}F_R\right)^{d_y-1}  \sqrt{2\uppi} \right). 
 \eeq
Therefore, combining \eqref{eq:Yellow_Box}, \eqref{eq:I} and \eqref{ineq:Final_I} we arrive at
\beq\label{ineq:h(y)_Cst}
\int_{\mY} h(y) \, \sd y 
\leq \frac{S_{d_y - 1}}{|R|^{\frac{1}{2}}} \, \mI 
\leq \frac{2^{d_y - 1} \uppi^{\frac{d_y}{2}}}{\abs{R}^{\frac{1}{2}}} + \frac{2\uppi^{\frac{d_y+1}{2}}(3F_R)^{d_y-1}}{\abs{R}^{\frac{1}{2}}\Gamma\left( \frac{d_y}{2} \right)}.
\eeq
Multiplying by \(\frac{C_R }{\abs{R}^{1/2}} \exp\left\{\tfrac{11}{4}F_R^2\right\}\) on both sides of \eqref{ineq:h(y)_Cst}, and simplifying, we obtain the upper bound
\beq\label{F_Bound} 
\frac{C_R }{\abs{R}^{1/2}}\exp\left\{\tfrac{11}{4}F_R^2\right\} \int_{\mY} h(y) \, \sd y 
\le \mathcal{K}_2,
\eeq
where
\[
\mathcal{K}_2 := \frac{\exp\left\{\frac{11}{4}F_R^2\right\}}{|R|^{\frac{3}{2}}} \left( 2^{\frac{d_y}{2} - 1} + \frac{2^{1 - \frac{d_y}{2}} \sqrt{\pi} (3F_R)^{d_y-1}}{\Gamma(\frac{d_y}{2})} \right)<\infty.
\]

\begin{remark}
Using \eqref{F_Bound}, together with \eqref{eq:h(y)}, \eqref{B_Ratio}, and Remark~\ref{rem:link_norm}, we obtain
\[
\mathbb{E}\left[\| \ell_Y \|_{L^2(\mix^{d_z})}^2\right] \leq \mathcal{K}_2,
\]
which implies, in particular, that \( \ell \in L^2(\mY; L^2(\mix^{d_z})) \). Moreover, the upper bound \( \mathcal{K}_2 \) is independent of \( d_z \), and hence the assumptions of Corollary~\ref{cor:IS_uniform} are satisfied for model \eqref{eqBoundedSensor}.
\end{remark}


\subsection{Pointwise bounded likelihood}\label{sub:BI_Noise}
Assume that the likelihood function \( g_y^{d_z}(x, z) \) satisfies the pointwise bound
\beq\label{gy_hk}
g_y(x, z) \leq h(x, z)\, k(y),
\eeq
where \( h : \mathcal{X} \mapsto [0, \infty) \) and \( k : \mathcal{Y} \mapsto [0, \infty) \) are measurable functions such that 
\beq
\label{eq__HK}
\|h\|_{\infty} \leq H < \infty
\quad \text{and} \quad
\|k\|_{L^1(\lambda)} \leq K < \infty,
\eeq 
with constants \( H \) and \( K \) independent of the dimension \( d_z \). 

Under assumption \eqref{gy_hk}, we obtain the bound
\beq
\label{eq__l}
\| \ell_y \|_{L^p(\mix^{d_z})}^2 \leq k(y)\frac{ \int h(x,z)\, g_y(x, z)\, \mix^{d_z}( \sd (x, z))}{ \left[ \int g_y(x, z)\, \mix^{d_z}( \sd (x, z)) \right]^2 }
\eeq
that holds for arbitrary $Y=y$. When the observation $Y$ is random, we can take expectations on both sides of \eqref{eq__l} and then simplify to obtain
\begin{align*}
\mathbb{E}\left[\| \ell_Y \|_{L^2(\mix^{d_z})}^2\right] 
&\leq \int_{\mathcal{Y}} k(y)  \frac{ \int h(x,z)\, g_y(x, z)\, \mix^{d_z}( \sd (x, z)) }{ \int g_y(x, z)\, \mix^{d_z}( \sd (x, z)) } \, \lambda(\sd y) \\
&= \int_{\mathcal{Y}} k(y)  \pi_y(h) \, \lambda(\sd y) \\
&\leq H \int_{\mathcal{Y}} k(y) \, \lambda(\sd y) = HK<\infty.
\end{align*}

\begin{remark}
In summary, if \eqref{gy_hk} and \eqref{eq__HK} hold, then the second moment $\mathbb{E}\left[\| \ell_Y \|_{L^2(\mix^{d_z})}^2\right]$ is finite and independent of $d_z$. As a consequence, $\ell \in L^2(\mY; L^2(\mix^{d_z}))$ and the assumptions in Corollary \ref{cor:IS_uniform} are fulfilled for model \eqref{gy_hk}.
\end{remark}


\section{Standard importance samplers}\label{sec:standard_random}

The analytical approach developed in Section~\ref{sec:General_Random} for nested importance samplers can also be applied to standard IS when the aim os to study the effect of the dimension \(d_x\) of the variable of interest over the approximation error bounds and there are no nuisance parameters (i.e., $d_z=0$). The same type of results obtained in Section~\ref{sec:General_Random} regarding $d_z$ can be established, for $d_x$ and and, mutatis mutandis the same types of special cases examined in Section~\ref{sec:Special_cases} remain valid.

We first present the (simplified) model for Bayesian inference, then discuss Bochner integrability for this model and, finally, provide error bounds. 

\subsection{Model}\label{model_std_IS}
Let \((\Omega, \Sigma, \mathbb{P})\) be a probability space, and let $X \colon \Omega \mapsto \mathcal{X} \subseteq \mathbb{R}^{d_x},
Y \colon \Omega \mapsto \mathcal{Y} \subseteq \mathbb{R}^{d_y},$
denote two multidimensional real r.v.'s. The r.v. \(X\) represents the \(d_x\)-dimensional state that we aim to estimate, while \(Y\) represents observed data of dimension \(d_y\). The model is fully specified by:
\begin{itemize}
  \item The prior probability law of the variable of interest $X$, which is denoted $\pi_0(\sd x)$. We assume the ability to generate random samples from this distribution.
  \item Given \(X = x\), the observation \(Y\) is characterized by a conditional pdf w.r.t. the Lebesgue measure, which we denote by \(l^{d_x}(y \mid x)\).  
We explicitly indicate the dependence on the dimension \(d_x\), and define the likelihood function as
\[
l_y^{d_x}(x) := c \, l^{d_x}(y \mid x), \quad c \in \mathbb{R}^+,
\]
which can be evaluated up to a proportionality constant $c<\infty$ independent of $X$.
\end{itemize}

By Bayes’ theorem, the posterior distribution of \(X\) given \(Y=y\) is
\[
\pi(\mathrm{d}x) = \frac{l_y^{d_x}(x)\pi_0(\sd x)}{\pi_0(l_y^{d_x})},
\]
where the normalizing constant (a.k.a. Bayes evidence) is
\[
\pi_0(l_y^{d_x}) = \int_{\mathcal{X}} l_y^{d_x}(x)\,\pi_0(\sd x).
\]
For clarity, we carry out our analysis for a simpler version of the general algorithm where we only need to assume the ability to sample (not evaluate) the prior $\pi_0$. The resulting scheme is displayed as Algorithm \ref{A_IS}. Hereafter we refer to this specific procedure as standard IS.  

\begin{algorithm} \label{A_IS}
Standard importance sampler. 
\begin{enumerate}
\item Draw $N$ i.i.d. samples $x^1, \ldots, x^N$ from the prior law $\pi_0$.
\item Compute normalised importance weights $w^{i} = \frac{l_y^{d_x}(x^i)}{\sum_{i=1}^N l_y^{d_x}(x^i)}$.
\item Output the random probability measure $\pi^{N}(\sd x) = \sum_{i=1}^N w^{i} \delta_{x^i}(\sd x)$
\end{enumerate}
\end{algorithm}
The next section is devoted to the calculation of upper bounds for the $L_p$ norms of the approximation error $\pi(f)-\pi^{N}(f)$, namely
\beq
\left\|
	\pi(f)-\pi^{N}(f)
\right\|_p = \mbE\left[
	\left|
		\pi(f)-\pi^{N}(f)
	\right|^p
\right]^{\frac{1}{p}}, \quad p \ge 1,
\nn
\eeq
where the real test function $f$ is assumed to belong to $B(\mX)$, the space of real‑valued, measurable, bounded functions on $(\mX,\mB(\mX))$.

\subsection{Bochner integrability}

Following the framework of Section~\ref{sec:General_Random}, we again assume that the observation \( Y \) is a r.v. generated according to the model described in Section~\ref{model_std_IS}. We assume that, for all $y\in \mY$, the likelihood function satisfies \(l_y^{d_x} \in L^2(\pi_0)\) and is normalized, i.e.,
\[
\int_{\mathcal{Y}} l_y^{d_x}(x) \, \lambda (\sd y) = 1, \quad \forall x \in \mathbb{R}^{d_x},
\]
for the Lebesgue measure \( \lambda \) on \( \mathcal{Y} \). Under these assumptions, \( g_y^{d_x}(x) \) coincides with the probability density of \( y \) given \( x \).

Let us consider the probability measure \(\eta \in \mathcal{P}(\mathcal{Y})\), defined as
\[
\eta(\mathrm{d}y) := \pi_0(l_y) \, \lambda(\mathrm{d}y),
\]
which corresponds to the marginal probability law of the observations induced by the model.

In this setting, we define the \textit{link function} that maps each observation to a normalized likelihood in the space \(L^2(\pi_0)\) as
\beq
\begin{array}{llll}
    \ell: &\mathcal{Y} &\mapsto &L^2(\pi_0)\\
    &y &\leadsto &\ell_y := \frac{l_y^{d_x}(\cdot)}{\pi_0(l_y^{d_x})}.
\end{array}\notag
\label{eq:link_Stand}
\eeq
For a test function \( \varphi \in L^2(\pi_0) \), the posterior random expectation is given by
\[
\pi_Y(\varphi) = \int_\mathcal{X} \varphi(x) \, \pi_Y(\sd x) = \frac{\int_\mathcal{X} \varphi(x)\, l_Y^{d_x}(x) \, \pi_0(\sd x)}{\int_\mathcal{X} l_Y^{d_x}(x') \, \pi_0(\sd x')}= \left\langle \varphi, \ell_Y \right\rangle_{L^2(\pi_0)}.
\]

\medskip
The same as in Section \ref{sec:General_Random}, the boundedness of \( \varphi \mapsto \mathbb{E}[\pi_Y^p(\varphi)] \) for \( \varphi\in L^2(\pi_0) \) is directly related to the Bochner integrability of the link function \( y \mapsto \ell_y \). This requires \( \ell \) to be strongly \( \eta \)-measurable (which holds under mild regularity assumptions see Appendix \ref{Apx:B_M}) and the finiteness of the norm integral
\beq \label{eq:K_Standard}
K_p^{d_x}:=\int_{\mathcal{Y}} \| \ell_y \|_{L^2(\pi_0)}^p \, \eta(\sd y) < \infty.
\eeq
For \( 1 \leq p \leq \infty \), we denote by \( L^p(\mathcal{Y}; L^2(\pi_0)) \) the Banach space of Bochner \( p \)-integrable functions \( \upsilon : \mathcal{Y} \mapsto L^2(\pi_0) \) w.r.t. the measure \( \eta \) (cf. \cite{tuomas2016analysis}).

\subsection{Error bounds}

We can state several results for the standard IS Algorithm \ref{A_IS} which are analogous to those obtained in Section \ref{sec:General_Random} for the nested importance sampler. The proofs follows the same arguments and are, in fact, simpler. 

\begin{theorem}\label{thm:6.1}
Assume that the link function satisfies \( \ell \in  L^p(\mathcal{Y}; L^2(\pi_0)) \) for \( p \in \{1,2\} \). Then, for all test functions \( \varphi \in L^2(\pi_0) \) and the assumed $p\in\{1,2\}$,
\[
\mathbb{E}[\pi_Y^p(\varphi)] \leq K_p^{d_x} \, \|\varphi\|_{L^2(\pi_0)}^p,
\]
where $K_p^{d_x}$ is the constant in \eqref{eq:K_Standard}.
\end{theorem}

\begin{theorem}\label{thm:6.2}
For all $y\in \mY$, let \( l_y^{d_x}(x)\in L^2(\pi_0) \) be a strictly positive conditional pdf. Then, for every \( f\in B(\mathcal{X})\) the following statements are equivalent:
\begin{enumerate}[label=\roman*)]
    \item There is a constant \( C_2^{d_x} < \infty \), independent of \( N \) and \( f \) such that
    \begin{equation*}
        \left\| \pi_Y(f) - \pi_Y^{N}(f) \right\|_2 \leq \frac{C_2^{d_x} \| f \|_\infty}{\sqrt{N}}.
    \end{equation*}
    \item  \( \ell \in L^2\big( \mathcal{Y}; L^2(\pi_0) \big) \).
\end{enumerate}
Moreover 
\[
C_2^{d_x} \leq B_2 K_2^{d_x},
\]
where \( B_2 \) is a constant independent of $d_x$.
\end{theorem}
The proof of Theorem \ref{thm:6.1} follows the same arguments as the proof of Theorem \ref{thm:4.1}, which is detailed in Appendix \ref{apx:Bochner-F}. The proof of Theorem \ref{thm:6.2} is analogous to the proof of Theorem \ref{thm:thIS2R} (see Appendix \ref{appThIS3}).

Theorem \ref{thm:6.2} provides a necessary and sufficient condition for the approximation errors of standard IS to converge at the rate \(\mathcal{O}(N^{-1/2})\) for a given (arbitrary large) $d_x$.  
Moreover, it offers a means to quantify the influence of the dimension \(d_x\) of the variable of interest on the approximation error bounds.  
This is made explicit in the following theorem.

\begin{theorem}\label{thn:IS_Stnd_P}
For all $y\in \mY$, let \( l_y^{d_x}(x) \in L^2(\pi_0)\) be strictly positive conditional pdf. Assume that \(\ell \in  L^p\big( \mathcal{Y}; L^2(\pi_0) \big) \) for some \( p \in \{1, 2\} \),  and assume that there exists a polynomial \( \widetilde P_{n,p}(d_x) \) of degree \(m\), possibly depending on \(p\), such that
\begin{equation*}
K_p^{d_x} \leq \widetilde P_{m,p}(d_x), \quad \forall d_x \in \mbN, 
\end{equation*} where $K_p^{d_x}$ is defined in \eqref{eq:K_Standard}.
 Then, for any \( f \in B(\mX) \) and the assumed $p\in\{1,2\}$, there is a polynomial \(P_{n,p}(d_x)\) of degree at most \(n\leq m\), possibly depending on \(p\), such that
\[
\left\| \pi_Y(f) - \pi_Y^{N}(f) \right\|_p 
\leq \frac{P_{n,p}(d_x)\, \|f\|_\infty}{\sqrt{N}}.
\]
\end{theorem}
The proof of Theorem \ref{thn:IS_Stnd_P} follows the same arguments shown in Appendix \ref{ap:T4.3} for the proof of Theorem \ref{thIS2R}.

\begin{remark}
Under the assumptions of Theorem \ref{thn:IS_Stnd_P}, the number of samples $N$ required to accurately approximate the posterior expectation \(\pi_Y(f)\) grows at most polynomially with the dimension \(d_x\) of the variable of interest. This contrasts with the typical exponential growth associated with high-dimensional problems, highlighting how specific structural features of the model can mitigate the curse of dimensionality in IS.

\end{remark}

\begin{corollary}\label{cor:IS_Std}
Under the assumptions of Theorem~\ref{thn:IS_Stnd_P}, let \( m = 0 \); i.e. assume there is $p\in\{1,2\}$ and a constant \( K_p < \infty  \) such that
\begin{equation*}\label{eq:U_Kdz_std}
\sup_{d_x \in \mathbb{N}} K_p^{d_x} \leq \mathcal{K}_p.
\end{equation*}
Then, for any \( f \in B(\mX) \), there exists a constant \( \mathcal{C}_p < \infty \), independent of \( N \), \( f \), and the dimension \( d_x \), such that
\[
\left\| \pi_Y(f) - \pi_Y^{N}(f) \right\|_p 
\leq \frac{\mathcal{C}_p \, \|f\|_\infty}{\sqrt{N}}.
\]
In particular, $\lim_{N \to \infty} \| \pi_Y(f) - \pi_Y^{N}(f) \|_p = 0$
uniformly over \(d_x\).
\end{corollary}

\begin{remark}
It can be shown (by following the same arguments in Section~\ref{sec:Special_cases}) that the assumptions of Corollary~\ref{cor:IS_Std}, in particular, the existence of a finite constant \(\mathcal{K}_p\) independent of \(d_x\) such that $K_p^{d_x} \leq \mathcal{K}_p,$
are satisfied by the same type of  models considered in Section~\ref{sec:Special_cases}. This, in turn, ensures the existence of approximation error bounds that remain uniformly bounded w.r.t. the dimension \(d_x\).
\end{remark}


\section{Conclusions} \label{sec:Discussion}

We have investigated the approximation errors of a nested importance sampling algorithm for numerical Bayesian inference in problems with nuisance variables. In particular, we have studied a setting where the goal is to approximate the posterior probability distribution of a $d_x$-dimensional r.v. $X$ given a $d_y$-dimensional observation $Y$ for a model that involves a $d_z$-dimensional nuisance variable $Z$. The nested importance sampler integrates $Z$ out numerically, and we have analyzed the dependence of the error bounds of the algorithm on the dimension $d_z$.

Let $\pi$ and $\pi^{N,M}$ denote the true posterior distribution of $X$ given $Y$ and its nested sampling approximation, respectively. In full generality, the upper bounds on the $L^2$ norm of the approximation error $|\pi(f)-\pi^{N,M}(f)|$ (for a bounded test function $f$) can increase exponentially with $d_z$, in agreement with existing results for different importance sampling schemes \cite{Bengtsson08,Snyder08,Snyder15,agapiou2017importance,chatterjee2018sample}. In this paper, we have introduced a refined analysis that relates the approximation error $|\pi(f)-\pi^{N,M}(f)|$ to the Bochner integrability of a {\em link function} $\ell$ that generates normalised likelihoods for the pair $(X,Z)$. We have proved that the bounds on $\|\pi(f)-\pi^{N,M}(f)\|_2$ are finite (and vanish as $\mO(N^{-\frac{1}{2}})$, with $N$ the number of samples, for arbitrary large $d_z$) if, and only if, the link function is Bochner integrable. 

Furthermore, the dependence of the error bounds on the dimension $d_z$ of the nuisance variables can be quantified in terms of the moments of the link function w.r.t. the marginal law of the observations. In particular, our analytical setting enables the identification of sufficient conditions for the $L^1$ and $L^2$ error bounds to depend polynomially on $d_z$, as well as sufficient conditions to ensure bounds that are uniform over $d_z$. Based on these general results, we have analyzed some special cases. One of these special cases corresponds to linear and Gaussian models, where we have proved that Bochner 2-integrability of the link function alone is sufficient to ensure uniform $L^2$ bounds over $d_z$. Some other special cases involve models with bounded or dominated observation functions. Let us note that the latter are conditions that are typically satisfied in practical applications, where, for example, sensors recording physical magnitudes have a bounded operational range. 

Finally, we have also shown that the same type of analysis, based on the notion of Bochner integrability, and the properties of the norm of the link function can be applied to a standard IS algorithm (with no nuisance parameters) in order to quantify the dependence of the $L^1$ and $L^2$ error bounds on the dimension $d_x$ of the r.v. of interest $X$. Hence the curse of dimensionality can also be alleviated in this setting depending on the structure of the models.




\begin{appendix}


\section{Use of importance functions} \label{sec:importance_sampling}
In this Appendix, we analyze the connection between Algorithm~\ref{A1} and Algorithm~\ref{A2} in the context of models with nuisance variables. The same analysis can be applied to the standard IS with the straightforward modifications.

Recall the model described in Section \ref{ssec:sIS2} with prior measure $\pi_0$, kernel $\kappa^{d_z}$ and likelihood function $g_y^{d_z}$, and let us refer to it as model $\mM$. Then, choose importance functions $\nu \gg \pi_0$ and $\tau \gg \kappa^{d_z}$ as in Algorithm \ref{A1}, and define the relative densities
\begin{equation}
\rho(x) := \frac{\sd\pi_0}{\sd\nu}(x) \quad \text{and} \quad \omega(x,z) := \frac{\sd\kappa^{d_z}}{\sd\tau}(x,z),
\label{eqDefRhoOmega}
\end{equation}
as well as the modified likelihood 
\begin{equation}
\widetilde g_y(x,z) := g_y^{d_z}(x,z) \omega(x,z) \rho(x).
\label{eqDefTildeG}
\end{equation}

Let us denote as $\widetilde \mM$ the model with prior measure $\nu(\sd x)$, Markov kernel $\tau(x,\sd z)$ and likelihood $\widetilde g_y(x,z)$. It is straightforward to verify that, given some fixed observation $Y=y$, $\mM$ and $\widetilde \mM$ have the same marginal likelihood. Indeed, if we construct 
\begin{equation}
\widetilde l_y(x) := \int_{\mZ} \widetilde g_y(x,z) \tau(x,\sd z)
\label{eqDefModL}
\end{equation}
then we readily see that
\begin{eqnarray}
\nu(\widetilde l_y) 
&=& \int_{\mathcal{X}} \int_{\mathcal{Z}} \widetilde{g}_y(x,z) \tau(x,\sd z) \nu(\sd x) \nonumber\\
&=& \int_{\mathcal{X}} \int_{\mathcal{Z}} g_y^{d_z}(x,z) \kappa^{d_z}(x,\sd z) \pi_0(\sd x) = \pi_0(l_y^{d_z}),\nonumber
\end{eqnarray}
where the first equality follows from \eqref{eqDefModL} and the second one is obtained from the definitions in \eqref{eqDefRhoOmega} and \eqref{eqDefTildeG}. Moreover, by a similar argument, we see that the modified model $\widetilde \mM$ also displays the same posterior law as the original model $\mM$, i.e., for any integrable test function $f:\mX \mapsto \mbR$ we obtain
\begin{equation}
\widetilde \pi(f) = \frac{\nu(\widetilde l_y f)}{\nu(\widetilde l_y)} = \frac{\pi_0(l_y^{d_z} f)}{\pi_0(l_y^{d_z})} = \pi(f).
\nonumber
\end{equation}
Therefore, model $\widetilde \mM$ can be interpreted as a reparametrisation of the original model $\mM$.

Finally, the standard nested importance sampling Algorithm \ref{A2} applied to model $\widetilde \mM$ yields the following steps:
\begin{enumerate}
\item Draw $x^1, \ldots, x^N$ i.i.d. from $\nu(\sd x)$.
\item For each $i = 1, \ldots, N$
    \begin{itemize}
    \item draw $z^{i,1}, \ldots, z^{i,M}$ i.i.d. from $\tau(x,\sd z)$, and 
    \item compute the likelihoods $\widetilde l_y^M(x^i) := \frac{1}{M}\sum_{j=1}^M \widetilde g_y(x^i,z^{i,j})$, with $j=1, \ldots, M$.
    \end{itemize}
\item Compute the weights $w^{i,M} \propto \widetilde l_y^M(x^i)$, $i=1, \ldots, N$.
\end{enumerate}
Using the definitions in \eqref{eqDefRhoOmega}, \eqref{eqDefTildeG} and \eqref{eqDefModL} it is straightforward to see that the procedure above matches Algorithm \ref{A1} exactly. Therefore, if we are interested in the performance of Algorithm \ref{A1} for model $\mM$, it is sufficient to analyze the performance of Algorithm \ref{A2} for model $\widetilde\mM$.




\section{Proof of Theorem \ref{thIS2}} \label{appThIS2}

\subsection{Preliminaries}
Our goal is to establish explicit upper bounds for the \( L_p \) norms of the approximation errors
\[
\| \pi(f)-\pi^{N,M}(f)  \|_p, \quad p \ge 1,
\]
for any test function \( f \in B(\mathbb{R}^{d_x}) \).

The following two propositions are straightforward to prove, hence we simply state them.

\begin{proposition}\label{MIP}
Let $\alpha_0, \beta_0 \in \mathcal{P}(S)$ and let $l, \tilde{l}$ be positive, real-valued functions on $S$ satisfying $\alpha_0(l) < \infty$ and $\beta_0(\tilde{l}) < \infty$. The probability measures $\alpha, \beta \in \mP(S)$ defined by
\[
\alpha(f) = \frac{\alpha_0(l f)}{\alpha_0(l)}, \quad \text{and} \quad \beta(f) = \frac{\beta_0(\tilde{l} f)}{\beta_0(\tilde{l})}, \quad \forall f\in \mB(S), 
\]
satisfies
\[
\left| \alpha(f) - \beta(f) \right| \leq \frac{1}{\alpha_0(l)} \left[ \left| \alpha_0(l f) - \beta_0(\tilde{l} f) \right| + \| f \|_{\infty} \left| \alpha_0(l) - \beta_0(\tilde{l}) \right| \right].
\]
\end{proposition}

\begin{proposition}\label{prop:Unbiased_E}
Let \( l_y^{d_z}, l_y^{M,d_z}\) be defined as in Section \ref{ssec:sIS2} and Section \ref{Algorithm}, respectively. For any bounded measurable function \( f \), we have
\[
\mathbb{E}[\pi_0^N(f \, l_y^{d_z})] = \pi_0(f \, l_y^{d_z}), \quad 
\mathbb{E}[\pi_0^N(f \, l_y^{M,d_z})] = \pi_0(f \, l_y^{d_z}).
\]
\end{proposition}

Next, we obtain a partial characterization of the approximation errors

\begin{lemma}\label{lem:L2bound}
Let $y\in\mY$ and \( p \in \{ 1, 2\} \) be fixed. Assume that \( g_y\in L^2(\mix^{d_z}) \). Then, for any \( f \in \mB(\mbR) \), there exists a constant \( B_p > 0 \) depending only on \( p \), such that

\begin{equation}\label{eq:E_y_fix}
 \mathbb{E}\left[ \left| \pi_0(f l_y^{d_z}) - \pi_0^N(f l_y^{d_z}) \right|^p \right]
 + \mathbb{E}\left[ \left| \pi_0^N\big(f [l_y^{d_z} - l_y^{M, d_z}] \big) \right|^p \right]
\leq \frac{B_p \|f\|^p_\infty}{N^{p/2}} \, \|g_y^{d_z}\|_{L^2(\mix^{d_z})}^p.
\end{equation}
\end{lemma}

\begin{proof} Throughout the proof, we suppress the superscript \( d_z \) for notational simplicity.

Let $y \in \mathcal{Y}$ and $p \in \{1, 2\}$ be fixed. Let $\{x^i\}_{i=1}^N$ be a sequence of i.i.d. r.v.'s from the distribution $\pi_0$, and let $\{z^{ij}\}_{j=1}^M$ be a sequence of i.i.d. r.v.'s from the conditional distribution $\kappa(x^i, \cdot)$ for each $i \in \{1, \dots, N\}$. We seek to obtain the bound in Eq.~\eqref{eq:E_y_fix} by an application of the Marcinkiewicz–Zygmund (M-Z) inequality. For this purpose, we define the following zero mean i.i.d.\ r.v.'s (see  Proposition~\ref{prop:Unbiased_E})
\[
U_{i,f} := \frac{1}{N} \left( \pi_0(f \, l_y) - f(x^i) l_y(x^i) \right), \quad 
U_{i,f}^M := \frac{1}{N} \left( f(x^i) l_y(x^i) - f(x^i) l_y^M(x^i) \right).
\]
Note that
\beq \label{eq:Uis}
\sum_{i=1}^N U_{i,f} = \pi_0(f l_y) - \pi_0^N(f l_y), \; \; \text{ and } \;\;
\sum_{i=1}^N U_{i,f}^{M} = \pi_0^N(f l_y) - \pi_0^N(f l_y^{M}).
\eeq
By the inequality
$\lvert x - y \rvert^p \leq (\lvert x \rvert + \lvert y \rvert)^p \leq 2^{p-1} \left( \lvert x \rvert^p + \lvert y \rvert^p \right),$
it follows that
\begin{equation}\label{eq:Uif_UiM}
U_{i,f}^2 \leq \frac{2 \|f\|^2_{\infty}}{N^2}\left(\pi_0^2(l_y) + l_y^2(x^i)\right), \quad
(U_{i,f}^M)^2 \leq \frac{2 \|f\|^2_{\infty}}{N^2}\left(l_y^2(x^i) + (l_y^M(x^i))^2\right). 
\end{equation}
In order to apply the M-Z inequality, for \(p = 1\) we first take the expectation using the reverse Jensen inequality, while for \(p = 2\) we compute the expectation directly. It is then straightforward to show that
\begin{align}
\sum_{i=1}^N \; \mathbb{E}[\; 
 U_{i,f}^2 \; ] &\leq \frac{2\|f\|^2_{\infty} }{N}\left(\pi_0^2(l_y)+\pi_0(l_y^2)\right), \label{eq:I_Ui}\\
\sum_{i=1}^N \mathbb{E}[(U^M_{i,f})^2] &\leq \frac{2\|f\|^2_{\infty} }{N}\left( \pi_0(l_y^2)+ \frac{1}{M} \left(\pi_0(\kappa(g_y^2)) + 2\pi_0(l_y^2) \right) \right). \label{eq:I_UiM}   
\end{align}

Indeed, to derive \eqref{eq:I_UiM}, we first observe that
\begin{equation*}
    \left( l_y^{M}(x^i) \right)^2 = \frac{1}{M^2} \left( \sum_{j=1}^M g_y^2(x^i, z^{ij}) + 2 \sum_{j \neq k}^M g_y(x^i, z^{ij}) \, g_y(x^i, z^{ik}) \right).
\end{equation*}
Define \( \mathcal{F} := \sigma\)-\( (X^i) \), the $\sigma$-algebra generated by the r.v.'s $X^i$. Then, by the tower property
\begin{align}
\mathbb{E} \left[\left( l_y^M (x^i) \right)^2 \right]
&= \mathbb{E} \left[ \frac{1 }{M^2} \mathbb{E} \left[  \sum_{j=1}^M g_y^2(x^i,z^{ij}) + 2\sum_{j\neq k}^M g_y(x^i,z^{ij}) g_y(x^i,z^{ik}) \; \middle| \; \mathcal{F} \right] \right] \notag \\
&= \frac{1}{M^2} \mathbb{E} \left[  \sum_{j=1}^M \kappa (g_y^2(x^i,\cdot)) + 2\sum_{j\neq k}^M l_y(x^i)^2 \right] = \frac{1}{M} \left(\pi_0(\kappa(g_y^2)) + 2\pi_0(l_y^2) \right), \label{eq:Ex_lyM}
\end{align}
where we have used the conditional independence of \( z^{ij} \) and \( z^{ik} \) for \( j \neq k \).

\noindent Observe that
\beq \label{des:pi}
0 < \pi_0^2(l_y) \leq \pi_0\big( l_y^2 \big) \leq \pi_0\big( \kappa(g_y^2) \big) = \|g_y^{d_z}\|_{L^2(\mix^{d_z})}^2 < \infty,
\eeq
hence, from \eqref{eq:I_Ui}, \eqref{eq:I_UiM} and \eqref{des:pi}
\beq \label{eq:UisB}
\sum_{i=1}^N \mathbb{E}\left[ U_{i,f}^2 \right] 
\leq \frac{4 \|f\|_\infty^2}{N} \, \|g_y^{d_z}\|_{L^2(\mix^{d_z})}^2, \; \; \text { and } \; \;
\sum_{i=1}^N \mathbb{E}\left[ \left(U_{i,f}^M\right)^2 \right] 
\leq \frac{8 \|f\|_\infty^2}{N} \, \|g_y^{d_z}\|_{L^2(\mix^{d_z})}^2.
\eeq
Using the upper bounds in~\eqref{eq:UisB}, we obtain~\eqref{eq:E_y_fix} by a direct application of the M--Z inequality to~\eqref{eq:Uis}.

\end{proof}

\begin{corollary}\label{cor:L2bound}
Suppose that \( \|g_y^{d_z}\|_{\infty} < \infty \). Then the conclusion of Lemma~\ref{lem:L2bound} holds for all \( p \geq 1 \), with the bound
\[
\mathbb{E}\left[ \left| \pi_0(f\, l_y^{d_z}) - \pi_0^N(f\, l_y^{d_z}) \right|^p \right]
+ \mathbb{E}\left[ \left| \pi_0^N(f [ l_y^{d_z} - l_y^{M, d_z} ] ) \right|^p \right]
\leq \frac{B_p\, \|f\|_\infty^p}{N^{p/2}} \, \|g_y^{d_z}\|_\infty^p.
\]
\end{corollary}

\begin{proof}
The proof follows the same steps as in Lemma~\ref{lem:L2bound}, with the only change occurring in the bound corresponding to \eqref{eq:Uif_UiM}, where we now use
\[
U_{i,f}^2 \leq \frac{2 \|f\|_{\infty}^2}{N^2} \left(2 \|g_y^{d_z}\|_{\infty}^2 \right), \qquad
(U_{i,f}^M)^2 \leq \frac{2 \|f\|_{\infty}^2}{N^2} \left(2 \|g_y^{d_z}\|_{\infty}^2 \right).
\]
\end{proof}

\subsection{Proof of Theorem \ref{thIS2}}

Note that
\[
\pi(f) - \pi^{M,N}(f) = \frac{\pi_0(f \, l_y^{d_z})}{\pi_0(l_y^{d_z})} - \frac{\pi_0^N(f \, l_y^{M,d_z})}{\pi_0^N(l_y^{M,d_z})}.
\]
Applying Proposition~\ref{MIP}, we obtain
\begin{equation} \label{eq:DI}
\left| \pi(f) - \pi^{M,N}(f) \right| 
\leq \frac{1}{\pi_0(l_y^{d_z})} \left[ 
\left| \pi_0^N(f \, l_y^{M,d_z}) - \pi_0(f \, l_y^{d_z}) \right| 
+ \| f \|_{\infty} \left| \pi_0(l_y^{d_z}) - \pi_0^N(l_y^{M,d_z}) \right| 
\right].
\end{equation}
Since the second term is a particular case of the first with \( f \equiv 1 \), we focus on bounding the first term. By the triangle inequality
\begin{equation} \label{eq:triangle}
\left| \pi_0(f \, l_y^{d_z}) - \pi_0^N(f \, l_y^{M,d_z}) \right| 
\leq \left| \pi_0(f \, l_y^{d_z}) - \pi_0^N(f \, l_y^{d_z}) \right| 
+ \left| \pi_0^N(f [l_y^{d_z} - l_y^{M,d_z}]) \right|.
\end{equation}
Raising both sides of \eqref{eq:triangle} to power \( p \), taking expectations, and applying Corollary~\ref{cor:L2bound}, we find
\begin{equation*} 
\| \pi(f) - \pi^{M,N}(f) \|_p^p 
\leq \left( \frac{\mathcal{C}_p}{\pi_0(l_y^{d_z})} \right)^p 
\frac{\| f \|_{\infty}^p}{N^{p/2}} 
\| g_y^{d_z} \|_{\infty}^p,
\end{equation*}
where \( \mathcal{C}_p<\infty \) is a constant depending only on \( p \). The inequality \eqref{eq:ThIS2} follows immediately, with 
\beq\label{eq:C_ydz}
C_y^{d_z}:= \frac{\mathcal{C}_p \| g_y^{d_z} \|_{\infty}}{\pi_0(l_y^{d_z})}<\infty.
\eeq 
\qed

\section{Linear Gaussian model with fixed observations}\label{Ap:Linear_Gaussian}
\sectionmark{Linear Gaussian model}
The aim of this section is to derive an upper bound for the constant $C_y^{d_z}$ in \eqref{eq:ThIS2} and \eqref{eqCy}, where the dependence on $d_z$ is explicit.  This requires obtaining a suitable lower bound for $\pi_0(l_y^{d_z})$. In light of Eq. \eqref{eq:pi0lyd}, we seek explicit bounds for 
$$ \abs{\Sigma_y^{d_z}} \;\; \text{ and } \;\; \exp\left\{
	-\frac{1}{2}(y-\mu_y^{d_z})^\top(\Sigma_y^{d_z})^{-1}(y-\mu_y^{d_z})
\right\}.$$ 

We recall some notation. For a general matrix $M$, let $\sigma_1(M)$ denote its largest singular value. If $M$ is square whit dimension $n\times n$, let $\mathrm{spec}({M})$ denote the set of its eigenvalues. If ${M}$ has real eigenvalues, then $\mathrm{spec}({M}) \subseteq [\lambda_n(M),\lambda_1(M)]$  where $\lambda_n(M),\lambda_1(M)$ denote the minimum and maximum eigenvalues of $M$, respectively.

\subsection{Preliminaries}

The determinant of $\Sigma_y^{d_z}$ can be bounded in terms of the eigenvalues and singular values of the model matrices using the following algebraic results.

\begin{lemma}\label{sumeig}
    Let $M_1$ and $M_2$ be real $n\times n$ symmetric matrices. Then $M_1,{M_2}$ and $M_1+{M_2}$ have real eigenvalues and
    \begin{equation*}
        \lambda_n(M_1)+\lambda_n(M_2) \leq \lambda_n(M_1+M_2) \leq \lambda_1(M_1+M_2) \leq \lambda_1(M_1)+\lambda_1(M_2).
    \end{equation*}
\end{lemma}

\begin{proof}
   It is a classical result that every real symmetric matrix has real eigenvalues. Let ${M}$ be a symmetric matrix with real eigenvalues, then the spectrum of the shifted matrix \( M^+ := M - \lambda_n(M) I \) satisfies $\mathrm{spec}({M}-\lambda_n(M) I)\subseteq[0,\lambda_1(M)-\lambda_n(M)]$.  Indeed, let \( v \neq 0 \) be an eigenvector of \( M^+ \) with associated eigenvalue \( \gamma \). We aim to show that \( 0\leq \gamma \leq \lambda_1(M) -\lambda_n(M)\). Since
\[
M^+ v = \gamma v = (M - \lambda_n(M) I) v = Mv - \lambda_n(M) v,
\]
it follows that
\[
Mv = (\gamma + \lambda_n(M)) v.
\]
That is, \( \gamma + \lambda_n(M) \in \mathrm{spec}(M) \), and therefore
\[
\gamma + \lambda_n(M) = \lambda_j(M)
\]
for some \( j \in \{1, \ldots, n\} \). This implies
\[
0\leq  \lambda_j(M) - \lambda_n(M) = \gamma \leq \lambda_1(M) - \lambda_n(M),
\]
as claimed.
    Similarly, one can show that the matrix ${M}^-:={M}-\lambda_1(M)I$ satisfies $\mathrm{spec}({M}^-)\in [\lambda_n(M)-\lambda_1(M),0]$.  
   Now, note that   
   $$M_1^+ +M_2^+={M_1}-\lambda_n(M_1)I + {M_2}-\lambda_n(M_2)I=({M_1}+{M_2})-(\lambda_n(M_1)+\lambda_n(M_2))I,$$ and the same argument applies. Consequently, $\mathrm{spec}(({M_1}+{M_2})-(\lambda_n(M_1)+\lambda_n(M_2))I)\subseteq [0,\infty)$, i.e., $\mathrm{spec}({M_1}+{M_2})\subseteq[\lambda_1(M_1)+\lambda_n(M_2),\infty)$. Similarly, the matrix  
    $$M_1^-+M_2^-={M_1}-\lambda_1(M_1)I + {M_2}-\lambda_1(M_2)I=({M_1}+{M_2})-(\lambda_1(M_1)+\lambda_1(M_2))I,$$
 leading to $\mathrm{spec}({M_1}+{M_2})\subseteq (-\infty, \lambda_1(M_1)+\lambda_1(M_2)].$ Combining both results, we conclude that $$\mathrm{spec}({M_1}+{M_2})\subseteq [\lambda_n(M_1)+\lambda_n(M_2),\lambda_1(M_1) +\lambda_1(M_2)].$$ 
 \end{proof}

 \begin{lemma}\label{deteig}
     Let ${M_1}$ and ${M_2}$ be  positive semidefinite $n\times n$ 
 Hermitian matrices.  Then 
 \begin{equation*}
     \prod_{i=1}^n(\lambda_i(M_1)+\lambda_i(M_2))\leq \abs{{M_1}+{M_2}}\leq \prod_{i=1}^n(\lambda_i(M_1)+\lambda_{n+1-i}(M_2)). 
 \end{equation*}
 \begin{proof}
    See Eq.~(2) of the Theorem in~\cite{fiedler1971bounds}.
 \end{proof}
 \end{lemma}

\begin{lemma}\label{334}
   Let ${M_1}\in \mathbb{R}^{m\times p}$ and ${M_2}\in \mathbb{R}^{p\times n}$ be given matrices, let $q:=\min\{m,p,n\}$, and denote the ordered singular values of ${M_1}$, ${M_2}$, and ${M_1}{M_2}$ by $\sigma_1({M_1})\geq\dots \geq \sigma_{\min\{m,p\}}({M_1})\geq 0 $, $\sigma_1({M_2})\geq\dots \geq \sigma_{\min\{p,n\}}({M_2})\geq 0 $, and $\sigma_1({M_1}{M_2})\geq\dots \geq \sigma_{\min\{m,n\}}({M_1}{M_2})\geq 0.$ Then  
   \begin{equation*}
       \prod_{i=1}^k \sigma_i({M_1}{M_2})\leq \prod_{i=1}^k \sigma_i({M_1})\sigma_i({M_2}), \hspace{.5cm} k=1,...,q.
   \end{equation*}
   If $n=p=m$, then the equality holds for $k=n$.
\end{lemma}
\begin{proof}
   See Theorem 3.3.4 in \cite{fiedler1971bounds}.
\end{proof}

\begin{lemma}\label{3316} Let ${M_1},$ ${M_2} \in \mathbb{R}^{m\times n}$ be given and let $q=\min\{m,n\}$. The inequality 
\begin{equation*}
    \abs{\sigma_i({M_1}+{M_2})-\sigma_i({M_1})}\leq \sigma_1({M_2}), \hspace{.2cm} \text{for } i=1,...,q,
\end{equation*}
holds for the decreasingly ordered singular values of ${M_1},$ ${M_2}$ and ${M_1}+{M_2}.$
    In particular, $$\sigma_1({M_1}+{M_2})\leq \sigma_1({M_2})+\sigma_1({M_1}).$$

\end{lemma}
\begin{proof}
    See Theorem 3.3.16 in \cite{horn1994topics}.
\end{proof}

\subsection{Bound for $\lvert\Sigma_y^{d_z} \rvert$}

To lighten the notation, we omit the superscript $d_z$ in this subsection, and write $\Sigma_y,T,B,Q$ and $H$ instead of $\Sigma_y^{d_z},T^{d_z},B^{d_z},Q^{d_z}$ and $H^{d_z}$.

\begin{lemma}\label{BDE}
    Let $\Sigma_y:=T\Sigma_xT^\top+BQB^\top+R$ as defined in \eqref{eq:mu_Sd}. The following inequality holds \beq
\abs{\Sigma_y}\leq [\lambda_1(T\Sigma T^\top)+\lambda_1(BQB^\top)+\lambda_1(R)]^{d_y} \notag
    \eeq 
\end{lemma}
\begin{proof}
    Define the matrix $\Gamma:= T\Sigma_xT^\top+BQB^\top$. Then $\Sigma_y=\Gamma+ R$ where both $\Gamma$ and $R$ are positive semidefinite, symmetric  $d_y\times d_y$ matrices. By Lemma \ref{deteig}, we have
    $$\abs{\Gamma+R}\leq \prod_{i=1}^{d_y}(\lambda_i(R)+\lambda_{n+1-i}(\Gamma))\leq \prod_{i=1}^{d_y}(\lambda_1(R)+\lambda_1(\Gamma)),$$
    where $\lambda_i(R), \lambda_i(\Gamma),$ $i=1,...d_y$ are the eigenvalues of $R$ and $\Gamma$ respectively, arranged in decreasing order.
    Now, since $T\Sigma_xT^\top$ and $BQB^\top$ are real, symmetric  $d_y\times d_y$ matrices, we can apply  Lemma \ref{sumeig} to obtain
    $$\lambda_1(\Gamma) \leq \lambda_1(T\Sigma_x T^\top)+\lambda_1(BQB^\top).$$
    Therefore, $\lvert \Sigma_y \rvert$ can be bounded by
$$\abs{\Sigma_y}=\abs{\Gamma+R}\leq \prod_{i=1}^{d_y}(\lambda_1(R)+\lambda_1(\Gamma))\leq \prod_{i=1}^{d_y}(\lambda_1(R)+\lambda_1(T\Sigma_xT^\top)+\lambda_1(BQB^\top)).$$
\end{proof}

\begin{lemma}\label{BES}
  Let the matrices $A$, $H$, $B$ and $Q$ be as defined in the linear model of Section \ref{Gauss_Model}, and let $T=A+BH$. Then, the following bounds hold
  \beq \notag
\lambda_1(T\Sigma_xT^\top)\leq (\sigma_1(A)+\sigma_1(B)\sigma_1(H))^2 \lambda_1(\Sigma_x)  \quad \text{and} \quad 
  \lambda_1(BQB^\top)\leq \sigma_1(B)^2 \lambda_1(Q). 
  \eeq  
\end{lemma}

\begin{proof}
   By construction $\Sigma_x$ and $Q$ are symmetric positive definite matrices. Using the spectral theorem, we can decompose them as 
   $$\Sigma_x=U_x \Lambda_x U_x^\top \quad \text{and} \quad Q= U_z \Lambda_z U_z^\top,$$
   where \(U_x\) and \(U_z\) are unitary matrices, and \(\Lambda_x\) and \(\Lambda_z\) are diagonal matrices containing the eigenvalues of \(\Sigma_x\) and \(Q\), respectively.  
Consequently,
\[
T \Sigma_x T^\top = \tilde{T} \tilde{T}^\top,  
\qquad  
B Q B^\top = \tilde{B} \tilde{B}^\top,
\]
where
\[
\tilde{T} = T U_x \Lambda_x^{1/2},  
\qquad  
\tilde{B} = B U_z \Lambda_z^{1/2}.
\]

   Recall that the singular values of $\tilde B$ are the square roots of the eigenvalues of the Gram matrix $\tilde B \tilde B^\top$; the same holds for $\tilde T$. Therefore, we obtain
   \beq \label{eq:T_B}
   \sigma_1(\tilde T)^2=\lambda_1(T\Sigma_xT^\top), \; \; \text{ and } \; \; \sigma_1(\tilde B)^2=\lambda_1(BQB^\top). 
   \eeq
   Now, by Lemma \ref{334}, 
   \beq\label{eq:B}
   \sigma_1(\tilde B) \leq \sigma_1(B)\sigma_1(U_z\Lambda_z^{1/2})= \sigma_1(B)\sigma_1(\Lambda_z^{1/2}),
   \eeq
  where we used the fact that singular values are invariant under unitary transformations, specifically
$$\sigma_1(U_z\Lambda_z^{1/2})=\sigma_1(U_z\Lambda_z^{1/2}I)=\sigma_1(\Lambda_z^{1/2})$$ with $I$ denoting the identity matrix. Similarly, using Lemma \ref{334} and Lemma \ref{3316},
   \beq\label{eq:T}
   \sigma_1(\tilde T) \leq \sigma_1(T^{d_z})\sigma_1(U_x\Lambda_x)=\sigma_1(A+BH)\sigma_1(\Lambda_x)\leq (\sigma_1(A)+\sigma_1(B)\sigma_1(H))\sigma_1(\Lambda_x^{1/2}).
   \eeq
The conclusion of Lemma~\ref{BES} follows from the combination of \eqref{eq:T_B}, \eqref{eq:B}, and \eqref{eq:T}.
\end{proof}

Finally, we establish an upper bound for $\abs{\Sigma_y}$.
\begin{lemma}\label{lem:detT}
    The determinant of the covariance matrix $\Sigma_y$ defined in \eqref{eq:mu_Sd} satisfies
    \beq 
    \abs{\Sigma_y} \leq \left[(\sigma_1(A)+\sigma_1(B)\sigma_1(H))^2 \lambda_1(\Sigma_x)+\sigma_1(B)^2 \lambda_1(Q)+\lambda_1(R)\right]^{d_y}. \notag
    \eeq
\end{lemma}
\begin{proof}
This a direct consequence of Lemma \ref{BDE} and Lemma \ref{BES}.    
\end{proof}

\subsection{Exponential term}

We aim to derive a bound for the exponential term in Eq. \eqref{eq:pi0lyd}.
\begin{lemma}\label{lem:expT}
Let $\Sigma_y$ be a symmetric positive definite matrix in $\mathbb{R}^{d_y\times d_y}$, $r>0$ and $y \in B_r(\mu_y)$ where $B_r(\mu_y):=\{y\in \mathbb{R}^{d_y} : \|y-\mu_y\|\leq r\}.$ Then, the following upper bound holds
$$\frac{1}{2} (y-\mu_y)^\top (\Sigma_y)^{-1}(y-\mu_y) \leq \frac{1}{2} d_y^3 r^2 (\lambda_n(\Sigma_y))^{-1} , \;\; \forall y \in B_r(\mu_y),$$
where $\lambda_n(\Sigma_y)$ denotes the minimum eigenvalue of the matrix $\Sigma_y$.
\end{lemma}
\begin{proof}
By the spectral theorem, the matrix \(\Sigma_y^{-1}\) admits the eigenvalue decomposition
\[
\Sigma_y^{-1} = U \Lambda^{-1} U^\top,
\]
where \(U\) is a unitary matrix and \(\Lambda^{-1}\) is a diagonal matrix whose entries are the (strictly positive) eigenvalues of \(\Sigma_y^{-1}\).
 Then
\begin{equation}\label{expw}
\begin{split}
     \frac{1}{2}(y-\mu_y)^\top\Sigma_y^{-1}(y-\mu_y) 
      = \frac{1}{2}(y-\mu_y)^\top U\Lambda^{-1}U^\top(y-\mu_y)=
    \frac{1}{2}{\norm{w}^2},
\end{split}
\end{equation}
where
 $w:=\Lambda^{-\frac{1}{2}}U^\top (y-\mu_y).$  By imposing the constraint $y \in B_r(\mu_y)$, we can establish an upper bound for~(\ref{expw}). To proceed, let $S=\Lambda^{-\frac{1}{2}}U^\top$, then $\Sigma_y=S^\top S$. Consequently, we can express $w_i$ as 
$$w_i=\sum_{j=1}^{d_y} s_{ij}(y_j-\mu_{y_j}), \quad i=1,...,d_y,$$
where $s_{ij}=[S]_{ij}$. This implies
\begin{equation} \label{wbound}
\abs{w_i} \leq \sum_{j=1}^{d_y} \abs{s_{ij}} \abs{y_j - \mu_{y_j}} \leq d_y \max_{j=1,\dots,d_y} \abs{s_{ij}} \abs{y_j - \mu_{y_j}}.
\end{equation}
Since $s_{ij}=u_{ij}\lambda_j^{-\frac{1}{2}}(\Sigma_y)$, where $u_{ij}=[U^\top]_{ij}$, it follows that $\abs{s_{ij}}=\abs{u_{ij}}\lambda_j^{-\frac{1}{2}}(\Sigma_y)$. Moreover, since $U$ is  unitary, we have $$\sum_{j=1}^{d_y} u_{ij}^2=1, \text{ hence } \;\abs{u_{ij}}\leq 1,$$
and, as a consequence 
\beq\label{eq:lamd_Sy}
\abs{s_{ij}}\leq \lambda_j^{-\frac{1}{2}}(\Sigma_y).
\eeq
Substituting~\eqref{eq:lamd_Sy} into~\eqref{wbound}, we obtain
$$\abs{w_i}\leq d_y \max_{j=1,...,d_y} \lambda_j^{-\frac{1}{2}}\abs{y_j-\mu_{y_j}}.$$ Now, by assumption, $\abs{y_j-\mu_{y_j}}\leq r$.  Squaring both sides of the inequality above we arrive at
$w_i^2\leq (r d_y)^2 \max_{j=1,...,d_y} \lambda_j^{-1}.$ Therefore, we obtain the bound
\begin{equation} \label{bound}
\norm{w}^2 \leq d_y^3 r^2 \max_{j=1,\dots,d_y} \lambda_j^{-1}(\Sigma_y).
\end{equation}
Using the identity
\begin{equation*} 
\max_{j=1,\dots,d_y} \lambda_j^{-1}(\Sigma_y) = \frac{1}{\min_{j=1,\dots,d_y} \lambda_j(\Sigma_y)},
\end{equation*}
and substituting it into \eqref{bound}, we obtain
$$\norm{w}^2\leq  d_y^3 r^2(\lambda_n(\Sigma_y))^{-1},$$
where $\lambda_n(\Sigma_y)$ denotes the minimum eigenvalue of $\Sigma_y$. The proof is complete upon invoking Eq.~\eqref{expw}.
\end{proof}

\begin{lemma}\label{lem:C9}
 Let the matrix $\Sigma_y$ be defined as in~\eqref{eq:mu_Sd}. Then, for any $r > 0$ and any $y \in B_r(\mu_y)$, the following inequality holds
$$\frac{1}{2} (y-\mu_y)^\top (\Sigma_y)^{-1}(y-\mu_y)\leq \frac{1}{2}d_y^3 r^2 (\lambda_n(R))^{-1},$$ 
where $ B_r(\mu_y):=\{y\in \mathbb{R}^{d_y} : \|y-\mu_y\|\leq r\}$ and $\lambda_n(R)$ denotes the smallest eigenvalue of $R$.    
\end{lemma}

\begin{proof}
From the definition of \(\Sigma_y\) and the positive definiteness of the matrices in \eqref{eq:mu_Sd}, we have
$$v^\top\Sigma_y v=v^\top(T^{d_z}\Sigma_x (T^{d_z})^\top+BQB^\top+R)v\geq v^\top R v>0, \quad \forall v \neq 0.$$
Consequently,
$$\lambda_n(R)=\inf_{\norm{v}=1}v^\top Rv\leq \inf_{\norm{v}=1}v^\top\Sigma_y v=\lambda_n(\Sigma_y),$$
The claim follows directly from Lemma~\ref{lem:expT}.
\end{proof}

\begin{remark}\label{rem:r_dz}
Recall that \(\mu_y^{d_z} = T^{d_z} \mu_x = (A + B^{d_z} H^{d_z}) \mu_x\).  
By sub-multiplicativity of the norm,
\[
\| \mu_y^{d_z} \|
\leq \big( \|A\| + \|B^{d_z} H^{d_z}\| \big) \| \mu_x \|
\leq \left( \lambda_{\max}(A) + \sigma_1(B^{d_z}) \, \sigma_1(H^{d_z}) \right) \| \mu_x \|.
\]
Thus, the dependence of \(\| \mu_y^{d_z} \|\) on the dimension \(d_z\) is entirely determined by the singular values of the model matrices \(B^{d_z}\) and \(H^{d_z}\).

\end{remark}

\begin{lemma}\label{lem:dmodel}
Let the family of models $\{ \mM^{d_z}\}_{d_z\in \mbN}$ be as described in Section \ref{Gauss_Model}. Then, for any $r>0$ and $\forall \;y\in B_r(\mu_y^{d_z})$ the term $\pi_0(l_y^{d_z})$ from Eq. \eqref{eq:pi0lyd}  satisfies 
\begin{align*}
\frac{1}{\pi_0(l_y^{d_z})} \leq \Bigg[ &
\big( \sigma_1(A) + \sigma_1(B^{d_z}) \sigma_1(H^{d_z}) \big)^2 \lambda_1(\Sigma_x) \\
& + \sigma_1(B^{d_z})^2 \lambda_1(Q^{d_z}) + \lambda_1(R)
\Bigg]^{d_y} \exp\left( \frac{1}{2} d_y^3 r^2 \lambda_n(R)^{-1} \right).
\end{align*}

where $\sigma_1(M)$, $\lambda_1(M)$ and $\lambda_n(M)$ denote the largest singular value, the maximum eigenvalue, and the minimum eigenvalue of the matrix $M$ respectively.  
\end{lemma}

\begin{proof}
   The proof is a direct consequence of Lemma~\ref{lem:C9} and Lemma~\ref{lem:detT}.
\end{proof}

\subsection{Proof of Theorem \ref{thm:LG_pol}}

 
 Assume that there exists a polynomial \(\widetilde{P}_m(d_z)\) of degree \(m\) such that
\[
\max \left\{ \sigma_1(B^{d_z})^{2d_y},\, \sigma_1(H^{d_z})^{2d_y}, \; \lambda_1(Q^{d_z})^{d_y} \right\}\leq \widetilde{P}_m(d_z), 
\quad \forall d_z \in \mbN.
\]
Take $r>0$ fixed, by Lemma~\ref{lem:dmodel}, and since all remaining model matrices are fixed and independent of \(d_z\), we have
\[
\frac{\mathcal{C}_p}{\pi_0(l_y^{d_z})} \leq P_{n,r}(d_z), \quad \forall y\in B_{r}(\mu_y^{d_z}),
\]
where \(P_{n,r}(d_z)\) is a polynomial in \(d_z\) of degree at most \(n\leq 2m\) depending only on the model parameters and \(r\), but independent of \(N\), \(M\), and \(f\). Consequently, using \eqref{eqCy}, the constant \(C_y^{d_z}\) satisfies
\[
C_y^{d_z} \leq P_{n,r}(d_z), \quad \forall y\in B_{r}(\mu_y^{d_z}).
\]
\qed


\section{Bochner integrability and expectation functionals
}\sectionmark{Bochner integrability \& expectation}\label{apx:Bochner}

In this appendix, we investigate the relationship between the Bochner integrability of the link function
\(\ell\), defined in Eq.~\eqref{eq:link}, and the functional and bilinear form introduced in Section~\ref{sec:Bochner_Int} for the posterior random measure \(\pi_Y\).

 For the map \(\ell\) to be Bochner integrable, it is necessary and sufficient that it is strongly \(\eta\)-measurable and that the integral of its norms is finite; that is,
\[
\int_{\mathcal{Y}} \|\ell_y\|_{L^2(\mix^{d_z})} \, \eta(\mathrm{d}y) < \infty.
\]
The following lemma guarantees strong measurability of \(\ell\), as defined in \eqref{eq:link}, for all models considered in this work.

\subsection[Bochner measurability of \(\ell(y)\)]{Bochner measurability of \(\ell(y)\)}\label{Apx:B_M}

\begin{lemma}\label{lem:S_M}
Let \( (\mathcal{Y}, \mathcal{B}(\mathcal{Y}), \eta) \) and \( (\mathcal{X}\times \mZ, \mathcal{B}(\mathcal{X}\times \mZ), \mix) \) be measure spaces, and let \( L^2(\mix) \) denote the Hilbert space of square-integrable functions w.r.t. \( \mix \). Assume that the function \( g : \mathcal{Y} \times \mathcal{X}\times \mZ \to (0, \infty) \) satisfies the following two conditions:
\begin{enumerate}[label=\roman*)]
    \item \( g \) is jointly measurable w.r.t. the product \(\sigma\)-algebra \( \mathcal{B}(\mathcal{Y}) \otimes \mathcal{B}(\mathcal{X \times \mZ}) \),
    \item for every \( y \in \mathcal{Y} \), the section \( x \mapsto g(y \mid x,z) \) belongs to \( L^2(\mix) \).
\end{enumerate}
Then the map \( \ell : \mathcal{Y} \to L^2(\mix) \) defined as
\[
\ell_y := \frac{g(y\mid  \cdot, \cdot)}{\int_{\mathcal{X}} g(y\mid  x,z) \, \mix (\sd( x, z))}.
\]
\( \ell \) is strongly \( \eta \)-measurable.
\end{lemma}

\begin{proof}
Since \( L^2(\mix) \) is a separable Hilbert space, by Pettis' Theorem (cf. Theorem~1.1 in \cite{pettis1938integration}), it suffices to show that \( \ell \) is weakly measurable. This is equivalent to proving that the map
\[
y \mapsto \langle \widetilde \varphi, \ell(y) \rangle_{L^2(\mix)} \quad \text{is measurable for all } \widetilde \varphi \in L^2(\mix).
\]
Computing the inner product
\[
\langle \widetilde \varphi, \ell(y) \rangle_{L^2(\mix)} = \frac{1}{\mix(g_y)} \int_{\mathcal{X}} \widetilde \varphi(x,z) g_y(x,z)  \, \mix (\sd( x, z)).
\]
Since the map \( (y, x,z) \mapsto \widetilde  \varphi(x,z) g_y(x,z) \) is jointly measurable (as the pointwise product of measurable functions), we can apply the Fubini-Tonelli Theorem (cf. \cite{rudin1987real}, Theorem 8.8) to conclude that the map
\[
y \mapsto \int_{\mathcal{X}} \widetilde \varphi(x,z) g_y(x,z)  \, \mix (\sd( x, z))
\]
is measurable. Moreover, since \( y \mapsto \mix(g_y) \) is measurable and strictly positive, it follows that
\[
y \mapsto \langle \widetilde \varphi, \ell(y) \rangle
\]
is measurable for all \( \widetilde \varphi \in L^2(\mix) \). Therefore, \( \ell \) is weakly measurable.
\end{proof}

\subsection[Bounded functionals]{%
 On Bochner integrability and the boundedness of expectation functionals}\label{apx:Bochner-F}
  
\addvspace{1em}
Let $\mathcal{H}$ be a linear normed space over the field of real numbers $\mathbb{R}$. A mapping $F: \mathcal{H} \rightarrow \mathbb{R}$ is called a functional if it is linear; that is, for all $\upsilon, \varsigma \in \mathcal{H}$ and $\alpha, \beta \in \mathbb{R}$,
$$F(\alpha\upsilon + \beta\varsigma) = \alpha F(\upsilon) + \beta F(\varsigma).$$
A functional $F$ is called bounded if there exists a positive finite constant $C$ such that
$$|F(\upsilon)| \leq C \|\upsilon\|$$
for all $\upsilon \in \mathcal{H}$. The norm of the functional $F$ is defined as the smallest such constant $C$, given by
$$\|F\| = \sup_{\upsilon \in \mathcal{H}, \upsilon \neq 0} \frac{|F(\upsilon)|}{\|\upsilon\|} = \sup_{\|\upsilon\|=1} |F(\upsilon)|.$$

A mapping $B: \mathcal{H} \times \mathcal{H} \rightarrow \mathbb{R}$ is called a bilinear form (or bifunctional) if it is linear in each argument separately.
A bilinear form $B$ is called symmetric if $B(\upsilon, \varsigma) = B(\varsigma, \upsilon)$ for all $\upsilon, \varsigma \in \mathcal{H}$.
A bilinear form $B$ is called bounded if there exists a positive finite constant $C$ such that
$$|B(\upsilon, \varsigma)| \leq C \|\upsilon\| \|\varsigma\|$$
for all $\upsilon, \varsigma \in \mathcal{H}$. The norm of the bilinear form $B$ is defined as
$$\|B\| = \sup_{\upsilon, \varsigma \in \mathcal{H}, \upsilon, \varsigma \neq 0} \frac{|B(\upsilon, \varsigma)|}{\|\upsilon\| \|\varsigma\|}= \sup_{\|\upsilon\|=1, \|\varsigma\|=1} |B(\upsilon, \varsigma)|.$$

A quadratic form $Q: \mathcal{H} \rightarrow \mathbb{R}$ is induced by a symmetric bilinear form $B: \mathcal{H} \times \mathcal{H} \rightarrow \mathbb{R}$ where
$$Q(\upsilon) = B(\upsilon, \upsilon)$$
for all $\upsilon \in \mathcal{H}$.
The quadratic form $Q$ is called bounded if there exists a positive finite constant $C$ such that
$$|Q(\upsilon)| \leq C \|\upsilon\|^2$$
for all $\upsilon \in \mathcal{H}$. The norm of the quadratic form $Q$ is defined as
$$\|Q\| = \sup_{\upsilon \in \mathcal{H}, \upsilon \neq 0} \frac{|Q(\upsilon)|}{\|\upsilon\|^2}= \sup_{\|\upsilon\|=1} |Q(\upsilon)|.$$

\begin{remark}
If \( \varphi \in L^2(\pi_0) \), then its constant extension \(\widetilde\varphi\) satisfies \(\widetilde\varphi \in L^2(\mix^{d_z})\). Furthermore, the Bochner space \( L^\infty(\mathcal{Y}; L^2(\mix^{d_z})) \) strictly generalizes \( L^2(\mix^{d_z}) \) via the constant embedding
\begin{align*}
L^2(\mix^{d_z}) &\mapsto L^\infty(\mathcal{Y}; L^2(\mix^{d_z})), 
\\ \widetilde \varphi &\leadsto \widehat{\varphi}(y,x,z) := \widetilde\varphi(x,z).
\end{align*}
This embedding is isometric and enables test functions in \(L^2(\pi_0)\) or \(L^2(\mix^{d_z})\) to be regarded as elements of \(L^\infty(\mathcal{Y}; L^2(\mix^{d_z}))\) by treating them as constant in \(y\).
\end{remark}

\begin{lemma}\label{lem:bounded_functional}
Let $\upsilon,\varsigma\in L^\infty\!\bigl(\mathcal{Y};L^2(\mix^{d_z})\bigr)$, and define the functional $F:L^\infty\!\bigl(\mathcal{Y};L^2(\mix^{d_z})\bigr) \mapsto \mbR,$ and the symmetric bilinear form $B: L^\infty\!\bigl(\mathcal{Y};L^2(\mix^{d_z})\bigr) \times L^\infty\!\bigl(\mathcal{Y};L^2(\mix^{d_z})\bigr) \mapsto \mbR$ as
\[
F(\upsilon) := \int_{\mathcal{Y}}
\bigl\langle \upsilon, \ell_y \bigr\rangle_{L^2(\mix^{d_z})} \, \eta(\mathrm{d}y),
\quad
B(\upsilon,\varsigma) := \int_{\mathcal{Y}}
\bigl\langle \upsilon, \ell_y \bigr\rangle_{L^2(\mix^{d_z})}
\bigl\langle \varsigma, \ell_y \bigr\rangle_{L^2(\mix^{d_z})} \, \eta(\mathrm{d}y).
\] Then:
\begin{enumerate}[label=\roman*)]
  \item $\ell\in L^1(\mY; L^2(\mix^{d_z}))$ if, and only if, $F$ is bounded, and
  \item $\ell\in L^2(\mY; L^2(\mix^{d_z}))$ if, and only if, $B$ is bounded.
\end{enumerate}
In both cases, the optimal bound is
\begin{equation}\label{eq:Kdz}
K_p^{d_z}
=\int_{\mathcal{Y}}
\bigl\|\ell_y\bigr\|_{L^2(\mix^{d_z})}^p\,\eta(\sd y), \quad p=1,2.
\end{equation}
\end{lemma}

\begin{proof}
We present the proof of (ii), case (i) is entirely analogous. 

Assume first that \( \ell \in L^2(\mathcal{Y}; L^2(\mix^{d_z})) \), and let $\upsilon,\varsigma\in L^\infty(\mY; L^2(\mix^{d_z}))$. Then by the Cauchy–Schwarz inequality, 
\begin{align*}
\left| \left\langle \upsilon_y, \ell_y \right\rangle_{L^2(\mix^{d_z})} \right| 
& \leq \|\upsilon\|_{L^\infty(\mY; L^2(\mix^{d_z}))}  \|\ell_y\|_{L^2(\mix^{d_z})}, \\
\left| \, \left\langle \varsigma_y, \ell_y \right\rangle_{L^2(\mix^{d_z})} \right| 
& \leq \| \, \varsigma  \|_{L^\infty(\mY; L^2(\mix^{d_z}))}  \|\ell_y\|_{L^2(\mix^{d_z})},
\end{align*}
and thus
\begin{align}
\abs{B(\upsilon,\varsigma)} &\leq \|\upsilon\|_{L^\infty(\mY;L^2(\mix^{d_z}))}\|\varsigma\|_{L^\infty(\mY;L^2(\mix^{d_z}))} \int_{\mathcal{Y}} \|\ell_y\|_{L^2(\mix^{d_z})}^2 \, \eta(\sd y) \notag\\ & =  K_2^{d_z} \|\upsilon\|_{L^\infty(\mY;L^2(\mix^{d_z}))}\|\varsigma\|_{L^\infty(\mY;L^2(\mix^{d_z}))}<\infty. \notag
\end{align}
This shows that the bilinear form is bounded. Conversely, assume there exists a finite constant \(C > 0\) such that
\[
|B(\upsilon, \varsigma)| \leq C \, \|\upsilon\|_{L^\infty(\mathcal{Y}; L^2(\mix^{d_z}))} \|\varsigma\|_{L^\infty(\mathcal{Y}; L^2(\mix^{d_z}))}, \quad \forall \, \upsilon, \varsigma \in L^\infty(\mathcal{Y}; L^2(\mix^{d_z})).
\]
In particular, the associated quadratic form
$Q(\upsilon) := B(\upsilon, \upsilon)$
is bounded by
\[
Q(\upsilon) \leq C \, \|\upsilon\|_{L^\infty(\mathcal{Y}; L^2(\mix^{d_z}))}^2, \quad \forall \upsilon \in L^\infty(\mathcal{Y}; L^2(\mix^{d_z})).
\]
To show that \(\ell \in L^2\big(\mathcal{Y}; L^2(\mix^{d_z})\big)\), it suffices (by Lemma~\ref{lem:S_M}) to verify that Equation~\eqref{eq:Kpdz} holds. To this end, define the map
\[
\quad y \mapsto \upsilon_y:= \frac{g_y^{d_z}(\cdot, \cdot )}{\|g_y^{d_z}\|_{L^2(\mix^{d_z})}}.
\]
Observe that \(\upsilon_y \in L^2(\mix^{d_z})\) for all \(y \in \mathcal{Y}\), with \(\|\upsilon_y\|_{L^2(\mix^{d_z})} = 1\). Hence,
\[
\upsilon \in L^\infty(\mathcal{Y}; L^2(\mix^{d_z})) \quad \text{and} \quad \|\upsilon\|_{L^\infty(\mathcal{Y}; L^2(\mix^{d_z}))} = 1.
\]
Moreover, for each \(y \in \mathcal{Y}\),
\[
\left\langle \upsilon_y, \ell_y \right\rangle_{L^2(\mix^{d_z})} = \|\ell_y\|_{L^2(\mix^{d_z})}.
\] Indeed, a direct computation yields
\[
\left\langle \upsilon_y, \ell_y \right\rangle_{L^2(\mix^{d_z})} = \frac{\int_{\mathcal{X}\times \mZ} \left(g_y^{d_z}(x,z) \right)^2 \mix (\sd( x, z))}{\|g_y^{d_z}\|_{L^2(\mix^{d_z})}  \pi_0(l_y^{d_z})}  = \frac{\|g_y^{d_z}\|_{L^2(\mix^{d_z})}}{\pi_0(l_y^{d_z})} = \|\ell_y\|_{L^2(\mix^{d_z})}.
\]
Therefore, for this particular choice of \( \upsilon \in L^\infty(\mathcal{Y}; L^2(\mix^{d_z})) \), we have
\[
Q(\upsilon) = \int_{\mathcal{Y}} \|\ell_y\|_{L^2(\mix^{d_z})}^2 \, \eta(\sd y) = K_2^{d_z} \leq C,
\]
hence, \( \ell \in L^2(\mathcal{Y}; L^2(\mix^{d_z})) \).
\end{proof}

\begin{remark}\label{rem:proof_cor}
  Note that, for all $\varphi \in L^2(\pi_0)$, we have 
  $$ \mathbb{E}[ \pi_Y(\varphi) ]=F(\widetilde \varphi),\; \; \text{ and } \; \; \mathbb{E}[ \pi_Y^2(\varphi) ] = Q(\widetilde \varphi).$$ 
\end{remark}


\section{Proof of Theorem  \ref{thm:thIS2R}}\label{appThIS3}

We  prove (ii) $\implies$ (i) first.

 Suppose that \( \ell \in L^2\big( \mathcal{Y}; L^2(\pi_0) \big) \). Bayes' theorem readily yields
\[
\pi_Y(f) - \pi_Y^{M,N}(f) = \frac{\pi_0(f \, l_Y^{d_z})}{\pi_0(l_Y^{d_z})} - \frac{\pi_0^N(f \, l_Y^{M,d_z})}{\pi_0^N(l_Y^{M,d_z})}.
\]
Applying Proposition~\ref{MIP},  we obtain 
\begin{equation} \label{DIR}
\left| \pi_Y(f) - \pi_Y^{M,N}(f) \right| 
\leq \frac{1}{\pi_0(l_Y^{d_z})} \left[ \left| \pi_0(f \, l_Y^{d_z}) - \pi_0^N(f \, l_Y^{M,d_z}) \right| 
+ \|f\|_\infty \left| \pi_0^N(l_Y^{M,d_z}) - \pi_0(l_Y^{d_z}) \right| \right].
\end{equation}
Since the second term on the right-hand side of \eqref{DIR} corresponds to the case \( f \equiv 1 \), we focus on bounding the first term. By the triangle inequality,
\begin{align*}
\frac{1}{\pi_0(l_Y^{d_z})} \left| \pi_0(f \, l_Y^{d_z}) - \pi_0^N(f \, l_Y^{M,d_z}) \right|
&\leq \frac{1}{\pi_0(l_Y^{d_z})} \left[ \left| \pi_0(f \, l_Y^{d_z}) - \pi_0^N(f \, l_Y^{d_z}) \right| 
+ \left| \pi_0^N(f(l_Y^{d_z} - l_Y^{M,d_z})) \right| \right].
\end{align*}
Raising to power \( p \), and using the independence of \( Y \) from the Monte Carlo samples \( \{x^i, z^{ij}\}_{i=1,...,N; j=1,...,M} \), we obtain
\begin{align*}
\int_{\mathcal{Y}} \frac{1}{\pi_0(l_y^{d_z})^p} 
\left| \pi_0(f \, l_y^{d_z}) - \pi_0^N(f \, l_y^{M,d_z}) \right|^p 
 \eta(\sd y)
\leq \int_{\mathcal{Y}} \frac{1}{\pi_0(l_y^{d_z}) ^p} \bigg(
&\mathbb{E} \left[ \left| \pi_0(f \, l_y^{d_z}) - \pi_0^N(f \, l_y^{d_z}) \right|^p \right] \\
&+ \mathbb{E} \left[ \left| \pi_0^N\big(f (l_y^{d_z} - l_y^{M,d_z})\big) \right|^p \right]
\bigg) \eta(\sd y).
\end{align*}
Finally, applying Lemma~\ref{lem:L2bound} to both terms inside the integral yields
\begin{equation*}
\int_{\mathcal{Y}} \frac{1}{\pi_0(l_y^{d_z})^p} 
\left| \pi_0(f \, l_y^{d_z}) - \pi_0^N(f \, l_y^{M,d_z}) \right|^p \, \eta(\sd y) \leq 
 \frac{\widetilde B_p \|f\|_\infty^p}{N^{p/2}} \int_{\mathcal{Y}} \| \ell_y \|_{L^2(\mix^{d_z})}^p \, \eta(\sd y),
\end{equation*}
where the final equality follows from Remark~\ref{rem:link_norm}. For $f\equiv 1$, a similar argument show that  
\begin{equation*}
\int_{\mathcal{Y}} \frac{1}{\pi_0(l_y^{d_z})^p} 
\left| \pi_0( \, l_y^{d_z}) - \pi_0^N( \, l_y^{M,d_z}) \right|^p \eta(\sd y)\leq 
 \frac{\widehat B_p }{N^{p/2}} \int_{\mathcal{Y}} \| \ell_y \|_{L^2(\mix^{d_z})}^p \,\eta(\sd y).
\end{equation*}
Therefore
$$\mathbb{E} \left[ \left| \pi_Y(f) - \pi_Y^{M,N}(f) \right|^p \right] 
\leq \frac{B_p \|f\|_\infty^p}{N^{p/2}} \int_{\mathcal{Y}} \| \ell_y \|_{L^2(\mix^{d_z})}^p\, \eta(\sd y).$$

Now we prove that (i) $\implies$ (ii).

Assume that Eq. \eqref{eq:Iff_error} holds. First note that 
\begin{align}
(\pi_Y(f) - \pi_Y^{N,M}(f))^2 
= \frac{1}{\pi_0(l_Y^{d_z})^2} \Big[ 
&\left( \pi_0(f l_Y^{d_z}) - \pi_0^N(f l_Y^{M,d_z}) \right)^2 \notag \\
&+ 2\, \pi_Y^{N,M}(f) \left( \pi_0(f l_Y^{d_z}) - \pi_0^N(f l_Y^{M,d_z}) \right) \left( \pi_0^N(l_Y^{M,d_z}) - \pi_0(l_Y^{d_z}) \right) \notag \\
&+ \left( \pi_Y^{N,M}(f) \right)^2 \left( \pi_0^N(l_Y^{M,d_z}) - \pi_0(l_Y^{d_z}) \right)^2 
\Big]. \label{eq:Nec_Con}
\end{align}
Now, taking expectations (and using the independence between the observation $Y$ and the random samples $\{x^i,z^{ij}\}_{1=1,...,N; j=1,...,M}$ over each term in the expansion, from Eq.~\eqref{eq:Nec_Con} we observe that
\begin{equation}\label{eq:71}
\int_{\mathcal{Y}} \frac{1}{\pi_0(l_y^{d_z})^2} \mathbb{E} \left[ \left( \pi_y^{N,M}(f) \right)^2 \left( \pi_0^N(l_y^{M,d_z}) - \pi_0(l_y^{d_z}) \right)^2  \right] \eta(\sd y) < \infty
\end{equation}
is a necessary condition for \( \mathbb{E}\left[ \left( \pi_Y(f) - \pi_Y^{N,M}(f) \right)^2 \right]<\infty \), for any bounded function $f$. In particular, let us choose $f$ that is positive and bounded away from $0$, i.e., there is some constant $s>0$ such that
\[
s \leq f(x), \quad \text{for all } x \in \mathcal{X}.
\]
Hence, from Eq. \eqref{eq:71} 
\begin{align*}
s^2 \int_{\mathcal{Y}} \frac{1}{\pi_0(l_y^{d_z})^2} \, 
\mathbb{E} \left[ \left( \pi_0^N(l_y^{M,d_z}) - \pi_0(l_y^{d_z}) \right)^2 \right] \eta(\sd y)
\leq \int_{\mathcal{Y}} \frac{1}{\pi_0(l_y^{d_z})^2} \, \mathbb{E} \bigg[ 
&\left( \pi_y^{N,M}(f) \right)^2 \\
&\times \left( \pi_0^N(l_y^{M,d_z}) - \pi_0(l_y^{d_z}) \right)^2 \bigg] \eta(\sd y),
\end{align*}
for this particular choice of $f.$
This implies that 
\beq\label{eq:MbF}
\int_{\mathcal{Y}} \frac{1}{\pi_0(l_y^{d_z})^2} \, 
\mathbb{E} \left[ \left( \pi_0^N(l_y^{M,d_z}) - \pi_0(l_y^{d_z}) \right)^2 \right] \eta(\sd y)<\infty.
\eeq
Let us consider the sequence of i.i.d. zero-mean (see Proposition \ref{prop:Unbiased_E}) r.v.'s
\[
U_{i,M} := \frac{1}{N} \left( \pi_0(l_y^{d_z}) - l_y^{M,d_z}(x^i) \right), \quad i=1,...,N.
\]
Then we have
\beq\label{eq:SU_iM}
\sum_{i=1}^N U_{i,M} = \pi_0(l_y^{d_z}) - \pi_0^N(l_y^{M,d_z}).
\eeq
Note that 
\begin{align*}
U_{i,M}^2 
&= \frac{1}{N^2} \left( \pi_0(l_y^{d_z})^2 
- 2 \pi_0(l_y^{d_z}) l_y^{M,d_z}(x^i) 
+ \left(l_y^{M,d_z}(x^i) \right)^2 \right),
\end{align*}
and
\begin{align}\label{eq:Ult_U_i}
\sum_{i=1}^N U_{i,M}^2 
&= \frac{1}{N} \left( 
\pi_0(l_y^{d_z})^2 
- 2 \pi_0(l_y^{d_z}) \pi_0^N(l_y^{M,d_z} ) 
+ \pi_0^N( (l_y^{M,d_z})^2 ) 
\right).
\end{align}
Taking expectations on Eq.~\eqref{eq:Ult_U_i}, using Eq.~\eqref{eq:Ex_lyM} and applying the M-Z inequality we arrive at
\begin{equation}\label{eq:L_MZ}
A_2 \, \mathbb{E} \left[ \sum_{i=1}^N U_{i,M}^2 \right] 
= \frac{A_2}{N} \left( 
\frac{1}{M} \left(\pi_0(\kappa(g_y^2)) + 2\pi_0((l_y^{d_z})^2) \right) - \pi_0(l_y^{d_z})^2 
\right)
\leq \mathbb{E} \left[ \left| \sum_{i=1}^N U_i \right|^2 \right], 
\end{equation}
where the constant $A_2<\infty$ does not depend on the distribution of the r.v.'s $U_{i,M}$. Dividing each term in \eqref{eq:L_MZ} by $\pi_0(l_y^{d_z})^2$, integrating over $\mathcal{Y}$ w.r.t. $\eta(\sd y)$, and using Eq.~\eqref{eq:SU_iM}, we obtain
\begin{align*}
\int_{\mathcal{Y}} \frac{A_2}{N \, \pi_0(l_y^{d_z})^2} \left( 
\frac{1}{M} \left( \pi_0(\kappa(g_y^2)) + 2\pi_0((l_y^{d_z})^2) \right) 
- \pi_0(l_y^{d_z})^2 \right) \,\eta(\sd y) \\
\leq \int_{\mathcal{Y}} \frac{1}{\pi_0(l_y^{d_z})^2} \,
\mathbb{E} \left[ \left( \pi_0^N(l_y^{M,d_z}) - \pi_0(l_y^{d_z}) \right)^2 \right]\, \eta(\sd y).
\end{align*}
Finally, note that $\pi_0(l_y^{d_z})^2 \leq \pi_0(\kappa(g_y^2))$, and since the right-hand side of the previous inequality is finite by \eqref{eq:MbF}, it follows that the following integral is also finite
\[ \int_{\mathcal{Y}} \frac{ 
\pi_0(\kappa(g_y^2)) 
}{\pi_0(l_y^{d_z})^2} \, \eta(\sd y)  
= \int_{\mathcal{Y}} \| \ell_y \|_{L^2(\mix^{d_z})}^2 \, \eta(\sd y) <\infty,
\]
(see Remark \ref{rem:link_norm}), this concludes the proof.
\qed

\section{Proof of Theorem \ref{thIS2R}}\label{ap:T4.3}

The result follows directly from the sufficiency argument in the proof of  Theorem \ref{thm:thIS2R} (i.e., (ii) $\implies$ (i)), which remains valid for both \(p = 1\) and \(p = 2\). Under the assumption that there exists a polynomial \(\widetilde{P}_{m,p}(d_z)\) bounding \(K_p^{d_z}\), the conclusion is immediate.
\qed

\end{appendix}

\begin{acks}[Acknowledgments]

\end{acks}

\begin{funding}
FG and JM acknowledge the partial support of the Office of Naval Research (award N00014-22-1-2647), Spain's {\em Agencia Estatal de Investigaci\'on} (award PID2021-125159NB-I00) and {\em Comunidad de Madrid} (IDEA-CM project, ref. TEC-2024/COM-89).
\end{funding}

\bibliographystyle{imsart-number} 

\bibliography{bibliografia}       

\end{document}